\documentclass[12pt,a4paper]{article}

\usepackage[latin1]{inputenc}
\usepackage{amsfonts, amsmath, amssymb, amsthm, thmtools}
\usepackage{color, enumitem, fancyhdr, float, layout, multirow, setspace, verbatim}
\usepackage{graphicx}
\usepackage{caption, subcaption}
\usepackage{cleveref, nameref}
\usepackage[colorlinks=true, urlcolor=blue, citecolor=blue, linkcolor=blue]{hyperref}
\usepackage[authoryear]{natbib}
\usepackage[all]{xy}
\usepackage{authblk}

\newtheorem{description_}{Description}
\newtheorem{definition}{Definition}
\newtheorem{example}{Example}
\newtheorem{lemma}{Lemma}
\newtheorem{theorem}{Theorem}

\declaretheorem[name=Restriction, refname={restriction, restrictions}, Refname={Restriction, Restrictions}]{restriction_}

\title{Compensating for the loss of a chance}
\author[1]{Rafael Stern\thanks{rafaelst@stat.cmu.edu}}
\author[1]{Joseph Kadane\thanks{kadane@stat.cmu.edu}}
\affil[1]{Department of Statistics, Carnegie Mellon University}

\date{\today}

\begin{document}
\maketitle 

\abstract{Civil liability for a lost chance applies to cases in which a tortious action changes the probabilities of the outcomes that can be obtained by the victim. A central point in the application of this type of liability is the valuation of damages. Despite the practical importance of the valuation of lost chances, the legal restrictions that guide it have rarely been discussed explicitly. In order to discuss these restrictions, we propose an abstract description of a lost chance case in which there are multiple possible outcomes and the victim can make a choice that affects these outcomes. Given this description, we propose six conceptual questions to guide the valuation of lost chances. We discuss alternative answers to these questions and present the formulas that derive from them. More specifically, we show that the main formulas that have been proposed for medical misdiagnosis cases are particular instances of the alternatives we discuss.}


\section{Current valuation of lost chances}
\label{introduction}

Civil liability for a lost chance applies to cases in which a tortious action changes the probabilities of the outcomes that can be obtained by the victim\footnote{Specifically, \citet{King1998}[p.495] argues that ``the loss-of-a-chance doctrine should operate when the following criteria are present: (1) the defendant tortiously failed to satisfy a duty owed to the victim to protect or preserve the victim's prospects for some favorable outcome; (2) either (a) the duty owed to the victim was based on a special relationship, undertaking, or other basis sufficient to support a preexisting duty to protect the victim's likelihood of a more favorable outcome, or (b) the only question was how to reflect the presence of a preexisting condition in calculating the damages for a materialized injury that the defendant is proven to have probably actively, tortiously caused; (3) the defendant's tortious conduct reduced the likelihood that the victim would have otherwise achieved a more favorable outcome; and (4) the defendant's tortious conduct was the reason it was not feasible to determine precisely whether or not the more favorable outcome would have materialized bur for the tortious conduct''.}. The element of chance is responsible for differences between typical tort cases and lost chance cases. For example, in a typical tort case one proves that, if not for the tortious action, a better outcome certainly would have occurred. By contrast, in a lost-chance case, one proves that, if not for the tortious action, there would be a higher probability of a given desirable outcome. This different type of proof has been applied to varying degrees in Brazil \citep{Silva2013}, England \citep{Smith1999}, France \citep{Viney1998}[p. 74], Italy \citep{Miceli2013}, Portugal \citep{Ferreira2013} and the USA \citep{Fischer2001, Koch2009}.  

A central point in the application of civil liability for a lost chance is the valuation of damages \citep{King1981}. The most commonly used method of valuation is that of proportional damage. This method applies to situations in which, had the tortious action not occurred, the victim would obtain either their current outcome or a better outcome with, respectively, probabilities $1-p_{0}$ and $p_{0}$. Also, the tortious action reduces the probability of a better outcome from $p_{0}$ to $0$. For example, this situation is found when a lawyer negligently misses a deadline for an appeal. According to the rule of proportional damage, if the difference between the values of the better and the current outcome is $\Delta v$, then the value of the lost chance is $p_{0} \cdot \Delta v$. 

The hypothetical situation described in the rule of proportional damage is often not found in practice. For example, suppose that a patient consults a physician about an infection in his leg. Also, the physician negligently misdiagnoses the infection and reduces the patient's probability of not having his leg amputated from $p_{0}$ to $p_{1}$ $(0 < p_{1} < p_{0})$. Since the patient might not have his leg amputated despite the tortious action, this situation is different from the one discussed in the rule of proportional damage.

The valuation of damages in the patient's case is more controversial than in the type of case described by the rule of proportional damage. \citet{Noah2005} identifies that at least three formulas have been used to value the lost chance in similar situations\footnote{\citet{Noah2005} presents the cases \citet{Herskovits1983}, \citet{Fulton2002} and \citet{Falcon1990}. Also, \citet{Rhee2013} discusses the same formulas and their application in \citet{Matsuyama2008} and \citet{McKellips1987}.}: $(p_{0}-p_{1}) \cdot \Delta v$, $\frac{p_{0}-p_{1}}{p_{0}} \cdot \Delta v$ and $\frac{p_{0}-p_{1}}{1-p_{1}} \cdot \Delta v$. In particular, Noah identifies that each of the formulas was used by one of the judges who decided \citet{Herskovits1983}.

The difference between the above formulas isn't merely a technical detail. For example, if a leg amputation caused a patient a damage of $\$100,000$ and a medical misdiagnosis reduced the probability of avoiding the amputation from $95\%$ to $90\%$, then the formulas would prescribe compensations of approximately \$5,000, \$5,250 and \$50,000. The tenfold increase from the lowest to the highest compensation value indicates that the formulas $(p_{0}-p_{1}) \cdot \Delta v$ and $\frac{p_{0}-p_{1}}{1-p_{1}} \cdot \Delta v$ should require substantially different justifications. Nevertheless, these justifications and the legal restrictions that guide them have rarely been discussed explicitly. As a result, it is hard to anticipate what compensation will be awarded in a lost chance case.

The lack of an explicit understanding of the legal restrictions that guide the valuation of lost chances also makes it hard to generalize the traditional formulas to more complex situations. For example, how would the formula's change if the victim had more than two possible outcomes? How would the formulas change if the victim had an active role and could make a choice that would influence his outcome? Both of these considerations are relevant in \citet{Matos2005}, a Brazilian case in which the victim participated in a game show. The participant could decide whether to address the question that was posed, and if she decided to answer, her answer may or may not have been correct. Therefore, there were three possible outcomes.

Here, we discuss the legal restrictions that apply to the valuation of lost chances. We start from an abstract case description that is sufficiently general to encompass \citet{Matos2005}. Next, we use this description to present conceptual questions that we believe are relevant to value a lost chance. By exploring alternative answers to these questions, we find justifications for each of the compensation formulas that have been used in medical misdiagnosis cases; each of these formulas is shown to be equivalent to combinations of qualitative answers to the conceptual questions. We also show how these qualitative answers can be used to adapt the formulas presented in medical misdiagnosis to more general situations. The rules for compensation obtained in these situations resemble the ``expected-value approach'' that is described in \citet{King1981}[p.1384].

Our approach to valuation is based on obtaining indemnification, that is, undoing the damage that was caused by the tortious action. There exist valid reasons why a compensation valuation should provide more than indemnity. For example, compensation can be used to deter or punish tortious actions. Also, indemnity disregards the risk and costs of taking a case to court. One might argue that these factors should be included in the compensation. Nevertheless, we study compensation valuations that strictly provide indemnity. We take this approach because we believe that, in order to determine compensation, the decision-maker can start from indemnity and then increase this value to obtain other goals.

In order to discuss indemnity, one must be able to value situations in which outcomes are uncertain. We perform these valuations using Subjective Expected Utility Theory (SEUT) \citep{Lindley1971}. SEUT expresses the value of situations with uncertain outcomes in terms of the value and the probability of each outcome. In the next Section, we discuss the elements that we require in order to apply SEUT to lost chance cases.

Afterwards, we discuss the valuation of lost chances that have many possible outcomes. In this setting, there are several possible translations of indemnity into the SEUT model. The choice between these translations raise three conceptual questions. First, how much factual information about the victim's outcome should be used to calculate compensation? Second, how should this information be used to determine the outcome had the tortious action not occurred? Third, how should the word ``indemnity'' be interpreted in the context of SEUT? By combining different answers to these questions, we obtain several possible compensation valuations. We illustrate these compensation valuations with abstract examples. In particular, we apply the compensation valuations to medical malpractice cases and contrast the result with the formulas discussed in \citet{Noah2005}.

Finally, we generalize the above abstract description and consider a case such that the victim could make choices that would affect his outcome. This type of case raises three new conceptual questions: At the time he made his choice, had the tortious action not occurred, how much information would the victim have available? In the previous scenario, how would the victim use the available information in order to make a choice? If the victim makes a wrongful choice, how is his compensation affected by his choice? The main effect of these questions is to establish which party is responsible for the loss derived from the victim making an ill advised choice after the tortious action occurs. We discuss possible presumptions on the victim's choice-making and relate these presumptions to compensation valuation. We illustrate the compensation valuations that we obtain with an application to \citet{Matos2005}.

\section{The elements of lost chance valuation}
\label{section:elements}

In this Section we discuss the elements that we require to value a lost chance. Specifically, we assume that all the elements that are presented in Description \ref{description:no_choice} have been established in Court.

\begin{description_}
 \label{description:no_choice}
 The trier of fact determines that, first, the tortfeasor committed a tortious action and, second, among all legally relevant possibilities in a set $\mathcal{O}$, an outcome ``o'' occurred. This is the factual scenario. Furthermore, in the counterfactual scenario, had the tortious action not occurred, the outcomes in $\mathcal{O}$ would occur with different probabilities. The trier of fact assigns a value $V(o) \in \mathbb{R}$ to each outcome $o \in \mathcal{O}$, assigns for each value $v \in \mathbb{R}$ the amount of money $M(v) \in \mathbb{R}$ that has value $v$ and decides on what information, $\mathcal{K}$, about the factual outcome can be used to calculate compensation.
\end{description_}

The symbol $\mathcal{O}$ stands for the set of all the relevant outcomes that the victim may experience. An outcome in $\mathcal{O}$ may occur in the factual scenario, in the counterfactual scenario or in both. These outcomes must be mutually exclusive and exhaustive. That is, one and only one of the outcomes can occur in each scenario. For example, in some medical malpractice cases the victim either dies or doesn't die. In this example, $\mathcal{O} = \{\text{victim dies}, \text{victim doesn't die}\}$. Requiring the outcomes to be mutually exclusive isn't a strong restriction. For example, if one considers that there are two relevant outcomes that can both occur, $A$ and $B$, one can construct $\mathcal{O}$ by taking the combinations of occurrences that involve $A$ and $B$, that is, $\mathcal{O}=\{\text{A and B},\text{A and not B},\text{not A and B},\text{not A and not B}\}$. By contrast, the requirement that the list be exhaustive often requires careful consideration. Although a list can be made exhaustive by adding an extra outcome that stands for all possibilities that previously weren't in the list, this move usually isn't useful. The new outcome usually is abstract and hard to value. For example, the outcome ``victim doesn't die'' is harder to value than multiple outcomes of the type ``victim survives for $t$ years with quality of life $q$''. Since SEUT requires a value for each outcome, the latter more concrete outcomes are preferable over the abstract former ones.

We also specified that $\mathcal{O}$ should contain only outcomes that are legally relevant for the case under consideration. This requirement indicates that $\mathcal{O}$ reflects a legal construction that isn't necessarily equivalent to a scientific account of the possible outcomes. For example, \citet{Jansen1999} presents an hypothetical case in which a person is robbed before going to a casino. Even though the victim would probably lose the robbed money in the casino, this is not legally relevant. In this case, the legally relevant outcome is that the victim was robbed of a given amount of money. Since this loss is a recoverable harm, it limits further questions about what might have happened to the money. Although this hypothetical case involves no uncertainty, similar considerations can occur in lost chance cases. For example, consider a medical malpractice followed by the death of a victim who had a high risk of committing suicide. Similarly to the robbery case, if the lost chance of survival is a recoverable damage, then it can limit the considerations of whether the victim would have committed suicide.

Besides determining $\mathcal{O}$, in Description \ref{description:no_choice} the trier of fact also specifies a value, $V$, for each outcome. Since the value of each outcome varies from person to person, it is necessary to state whose valuation is represented by $V$. For an example of a lost chance case that illustrates the impact of valuating outcomes from different perspectives, we refer the reader to \citet{Aumann2003}. In the discussion of this example, Michael Keren argues that the valuation perspective should fit the intended effect of compensation: ``If punishment is `retribution', then the yardstick should be the effect on the victim; if `deterrent', then it should be the effect on the perpetrator'' \citep{Aumann2003}[p.237]. Since we focus on finding a compensation that indemnifies the victim, $V$ should be evaluated from the victim's perspective. For each possible outcome, $o$, the trier of fact provides a value for $o$, $V(o)$, by considering himself in the victim's place.

The calculations we perform in this paper require $V$ to be an utility function. Qualitatively, the higher the value of $V(o)$, the more $o$ is desirable. In this sense, $V(o)$ resembles a monetary valuation of $o$. The main difference between monetary value and utility is that, while the same increase in money can bring different values\footnote{For example, \citet{Bernoulli1954}[pp.24] argues that ``\ldots the determination of the value of an item must not be based on its price, but rather on the utility it yields. The price of the item is dependent only on the thing itself and is equal for everyone; the utility, however, is dependent on the particular circumstances of the person making the estimate''. In order to illustrate this difference, Bernoulli presents a gamble in which a person earns either a large sum or money with probability 0.5 or nothing. Prospective buyers of this gamble would be willing to pay different prices according to their aversion to risk. Although the price to participate in this gamble might be the same for all prospective buyers, their utilities for the gamble are different.}, the same increase in utility always brings the same gain. To the best of our knowledge, this distinction usually isn't made apparent in legal decisions. Most decisions use monetary values as a basis for calculating compensation. By using this methodology, one cannot incorporate a common phenomenon such as the tendency to avoid risks and, therefore, compensation values won't necessarily provide indemnification.

Although $V$ isn't a monetary valuation, the compensation value determined by the court must be expressed in terms of money. Given this constraint, our calculations also require the trier of fact to establish, for each value $v \in \mathbb{R}$, the amount of money that the trier of fact believes would bring an increase in value of $v$ to the victim. This element is presented in Description \ref{description:no_choice} by the function $M(v)$.

One can inspect the difference between $V(o)$ and $M(v)$ by contrasting the way in which each is used to determine compensation. Initially, $V(o)$ is used to determine the value of the lost chance in terms of the value of each possible outcome. For example, if the victim loses the chance to avoid pain and suffering, then $V(o)$ represents the victim's preferences over his emotional and physiological responses in each outcome in $\mathcal{O}$. The judge determines the value of the lost chance in terms of the value of each outcome. Only at a later stage would $M(v)$ be used to determine a monetary compensation in terms of the value of the lost chance.

The next element in Description \ref{description:no_choice} is the trier of fact's representation of uncertainty through probabilities. These probabilities reflect the trier of fact's uncertainty at a moment before the tortious action occurred. The probability model for Description \ref{description:no_choice} can be presented using an influence diagram\footnote{Formally, the influence diagram presents a twin-network model \citep{Balke1994}.} such as in Figure \ref{figure:description_no_choice}. The random quantities $O_{0}$ and $O_{1}$ stand for the outcomes that occur in the factual and counterfactual scenarios. Similarly, $V_{0}$ and $V_{1}$ stand for the values attributed to the outcome in the counterfactual and factual scenarios; $V_{0} = V(O_{0})$ and $V_{1} = V(O_{1})$. Finally, $F$ stands for the factors that connect the outcomes in each scenario. The marginal probabilities associated to $O_{0}$ and $O_{1}$ are often presented by expert witnesses in Court. Example \ref{example:medical_malpractice} uses a typical medical malpractice case to illustrate the model in Figure \ref{figure:description_no_choice}. For more uses of probability models in lost chance cases, we refer the reader to \citet{Miller2005, Pan2013}.

\begin{figure}[ht]
 \begin{align*}
  \xymatrix{
    *+[Fo]{O_{0}} \ar[d]	& *+[Fo]{F} \ar[l] \ar[r] 	& *+[Fo]{O_{1}} \ar[d] \\
    *+[F-]{V_{0}}		&				& *+[F-]{V_{1}}
  }
 \end{align*}
 \caption{Influence diagram for Description \ref{description:no_choice}. }
 \label{figure:description_no_choice}
\end{figure}
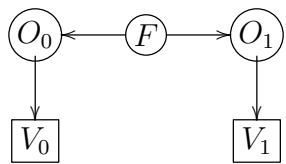

\begin{example}[typical medical malpractice case]
 \label{example:medical_malpractice}
 The victim can obtain either a bad outcome, $o_{b}$, or a good outcome, $o_{g}$. The difference in value between these outcomes is $\Delta v = V(o_{g})-V(o_{b}) > 0$. By performing an act of medical malpractice, the physician reduces the probability of the good outcome from $P(O_{0}=o_{g}) = p_{0}$ to $P(O_{1}=o_{g}) = p_{1}$. It is commonly established that, if the bad outcome were to occur had there been no medical misdiagnosis, then it would also occur in the event of the medical misdiagnosis. This relationship between each scenario is expressed through the influence of $F$ on $O_{0}$ and on $O_{1}$. For example, one can model this relation by assuming $F \sim \text{Uniform}(0,1)$, $O_{0}=o_{g}$ if and only if $F < p_{0}$ and $O_{1}=o_{g}$ if and only if $F < p_{1}$. 
\end{example}

The connection between the factual and counterfactual scenarios in Example \ref{example:medical_malpractice} is based on the assumption that appropriate diagnosis would never worsen the outcome obtained in the factual scenario. Evidence doesn't always corroborate this assumption. Examples \ref{example:factual-independence} and \ref{example:factual-dependence} present two similar situations that stress the importance of $F$, the connection between the factual and the counterfactual scenarios.

\begin{example}
 \label{example:factual-independence}
 The tortfeasor has two boxes that have red and blue balls. The relative frequencies of balls in boxes 1 and 2 are, respectively, $p_{0}$ and $p_{1}$ blue balls and $1-p_{0}$ and $1-p_{1}$ red balls, $p_{1} < p_{0}$. The tortfeasor had the obligation to randomly select a ball from box 1. The tortfeasor would then have to reward the victim according to the color of the ball that was drawn using the following scheme: V(red)=$v_{\text{r}}$, V(blue)= $v_{\text{b}}, v_{\text{b}}-v_{\text{r}}=\Delta v > 0$. Nevertheless, the tortfeasor secretly draws a ball from box 2, the ball that is drawn is red and the tortfeasor pays $v_{r}$ to the victim. There is evidence that a draw from box 1 would be independent from a draw from box 2.
\end{example}

\begin{example}
 \label{example:factual-dependence}
 The tortfeasor had the same obligation as in Example \ref{example:factual-independence}. Nevertheless, in this case, before a ball is drawn, the tortfeasor secretly paints red some of the blue balls in box 1. As a consequence, box 1's new relative frequencies were $p_{1}$ blue balls and $1-p_{1}$ red balls. Next, the tortfeasor removes a ball from box 1, the ball that is drawn is red and the tortfeasor pays $v_{r}$ to the victim.
\end{example}

From a perspective based on evidence, the connection between the factual and counterfactual scenarios is different in Examples \ref{example:factual-independence} and \ref{example:factual-dependence}. In Example \ref{example:factual-independence}, since the draws from each box are independent, the knowledge that a red ball was drawn from box 2 brings no information about which ball would have been drawn from box 1. Consequently, a blue draw in the factual scenario can occur jointly with a red draw in the counterfactual scenario. By contrast, in Example \ref{example:factual-dependence} the factual outcome brings information about the counterfactual outcome. No ball in box 1 was painted blue and, consequently, upon observing a blue ball, one knows for sure that a blue ball also would have been drawn in the counterfactual scenario. 

Despite the aforementioned difference between Examples \ref{example:factual-independence} and \ref{example:factual-dependence}, another line of inquiry approximates these scenarios. In both of them, the obligation of drawing a ball from a box with a relative frequency of $p_{0}$ blue balls was replaced by a draw from a box with a relative frequency of $p_{1}$ blue balls. Furthermore, even in Example \ref{example:factual-independence} one can associate the blue ball in the factual scenario to the blue ball in the counterfactual scenario. This type of inquiry questions whether there is any reason to compensate the victim differently in these examples. In other words, similarly to the legal construction of causality, the relation between the factual and counterfactual scenarios can be a legal construction that diverges from the scientific perspective. For example, one might adopt the connection that brings the outcomes in the factual scenario the closest to the outcomes in the counterfactual scenario. In Example \ref{example:factual-independence}, this type of reasoning is a possible justification for approximating a blue ball drawn from box 1 and a blue ball drawn from box 2, although evidence shows these draws are independent. In the next Section we discuss in more detail the differences in compensation that arise by adopting different connections between the factual and counterfactual scenarios. 

The last element of compensation presented in Description \ref{description:no_choice} is $\mathcal{K}$, the information about the factual outcome that can be used to determine compensation\footnote{Formally, the random variables $F$,$O_{0}$,$O_{1}$,$V_{0}$ and $V_{1}$ are defined in a probability space $(\Omega,\mathcal{F},\mathbb{P})$. We use the symbol $\mathcal{K}$ to denote a sub-$\sigma$-field of $\mathcal{F}$. In order to define this sub-$\sigma$-field, the $\sigma$ notation is often useful. For each $\mathcal{F}_{0} \subset \mathcal{F}$, we denote by $\sigma(\mathcal{F}_{0})$ the smallest $\sigma$-field that contains $\mathcal{F}_{0}$. Similarly, for every $\mathcal{F}$-measurable random variable, $Z$, we define $\sigma(Z) = \sigma(\{Z^{-1}[(a,\infty)]: a \in \mathbb{R}\})$.}. Given the uncertainty that exists in lost chance cases, the trier of fact might have more information about the factual outcome at the time of litigation than the tortfeasor had available at the time of the tortious action. Due to this asymmetry, during litigation the trier of fact might be certain about a given factual outcome that was very unlikely and couldn't reasonably be expected at the time the tortious action occurred. This type of situation leads \citet{Fisher1990} to question whether the trier of fact can use the full extent of the available factual information in order to determine compensation. The choice of $\mathcal{K}$ allows one to take this type of consideration into account and restrict the use of factual information. If two outcomes differ only in respect to inadmissible factual information according to $\mathcal{K}$, then their respective compensations should be equal. 

Although Description \ref{description:no_choice} assumes that the factual outcome is known at the time of trial, a change on $\mathcal{K}$ can be used to extend Description \ref{description:no_choice} to cases in which only partial information about the factual outcome is known at that time. In such cases, $\mathcal{K}$ is determined both by the legal restrictions and the evidence that is available\footnote{Formally, consider that $O_{1}^{t}$ represents the outcomes that are known at the time of the trial and that $O_{1}$ represent all the legally relevant outcomes that the victim will experience. In cases such that not every factual outcome is known at the time of trial, $\mathcal{K} \subset \sigma(O_{1}^{t}) \subset \sigma(O_{1})$.}. We believe that such an extension is of practical importance in some cases. For example, consider a medical misdiagnosis case such that the victim has an increased risk of suffering many different symptoms throughout the rest of his life. Since law suits are expensive, it is usually unreasonable to expect that the victim will file more than a single law suit. Also, in case the victim files a single law suit, then in case the extension to Description \ref{description:no_choice} isn't available, the victim won't be compensated for the symptoms that occur after the law suit is resolved. Therefore, in these cases, Description \ref{description:no_choice} would be a practical limitation to the victim's right to be fully compensated. The anticipation of future damages that we suggest by extending Description \ref{description:no_choice} avoids this practical limitation by allowing the victim to be fully compensated in a single law suit\footnote{Here, we disagree with the statement that ``where the defendant's tortious conduct created a risk of future consequences, the operation of the loss-of-a-chance doctrine should be suspended until the harmful effects actually materialize'' \citep{King1998}[p.496].}.

In the next Section we use the above elements to discuss the valuation of lost chances with many possible outcomes. We discuss several interpretations of how compensation should be determined. These interpretations are answers to conceptual questions such as: How much factual information can be used in order to determine compensation? For each set of interpretations, we obtain a compensation valuation in terms of the elements described in this Section. We contrast these compensation valuations by applying them to cases such as the medical malpractice case described in Example \ref{example:medical_malpractice}. For this example, we also contrast the compensation valuations to the formulas for compensation discussed in \citet{Noah2005}.

\section{Valuation of lost chances with many possible outcomes}
\label{section:no_choice}

In order to derive the compensation in Description \ref{description:no_choice}, we first translate legal restrictions on compensation into the model described in the previous Section. The legal restrictions we discuss in this Section are varying answers to the following conceptual questions:

\begin{enumerate}
 \item How much information about the factual outcome should be used to determine compensation?
 \item Is the connection between the factual scenario and the counterfactual scenario a matter of evidence or law?
 \item How should the expression ``undo the damage'' be interpreted in the context of lost chances?
\end{enumerate}

Once the above questions are discussed, we use the probability model and the legal restrictions on compensation to obtain rules for valuating lost chances.

In order to translate legal restrictions on compensation into the probabilistic model described in Section \ref{section:elements}, we first translate the word compensation. This translation is presented in Definition \ref{def:compensation}. 

\begin{definition}[compensation valuation]
 \label{def:compensation}
 A compensation valuation, $X$, is a function of the factual outcome that determines the value of the compensation that should be awarded. The compensation valuation is non-negative, since the legal decision cannot fine the victim. Also, the compensation valuation can use only the information allowed by $\mathcal{K}$\footnote{Formally, $\mathcal{K} \subset \sigma(O_{1})$ and $X$ is a non-negative $\mathcal{K}$-measurable random variable.}. We use the letter $\mathcal{X}$ to denote the set of all compensation valuations. The amount of money awarded to the victim is that amount which increases the victim's value by $X$\footnote{Formally, the victim should be awarded $M(V_{1}+X)-M(V_{1})$.}.
\end{definition}

The first legal restriction we discuss regards to how much information about the factual outcome can be used in order to determine compensation. This restriction is translated into the model in Section \ref{section:elements} as a determination of $\mathcal{K}$. On the one hand, one can argue that the use of factual information increases the uncertainty about the compensation value and, therefore, can make the tortfeasor owe values that are much higher than the one he would expect to owe at the moment he committed the tortious action. On the other hand, one can also argue that the use of factual information allows the victim to be compensated proportionally to the damage he effectively suffered. In the following, we present three possible restrictions that allow varying uses of factual information. 

\begin{restriction_}[L-FI]
  \label{restriction:low_information}
  Low Factual Information. No information about the factual outcome can be used\footnote{Formally, $\mathcal{K}=\sigma(\emptyset)$.}. That is, the victim receives the same compensation no matter what factual outcome occurred.
\end{restriction_}

Restriction \nameref{restriction:low_information} is the most extreme restriction among the possible uses of factual information. Since no factual information is used to determine compensation, at the moment of the occurrence of the tortious action the tortfeasor can determine exactly what compensation he will be required to pay in the future. Although Restriction \nameref{restriction:low_information} is defended in the case discussed in \citet{Fisher1990}, it has rarely been observed in lost chance cases. For example, consider the typical case of medical malpractice in Example \ref{example:medical_malpractice}. Restriction \nameref{restriction:low_information} determines that the compensation is always the same, no matter if the bad outcome occurs or not. On the contrary, the most common approach is to allow compensation only when the bad outcome occurs.

Other Restrictions compensate only the victims that obtain outcomes that are worse than the one they were expected to obtain in the counterfactual scenario. For example, in Example \ref{example:medical_malpractice}, only the victims that obtained the bad outcome would be compensated. This separation between factual outcomes is found in Definition \ref{def:outcome_separation}.

\begin{definition}[Selective compensation groups]
  \label{def:outcome_separation}
  We denote by $\mathcal{O}^{-}$ and $\mathcal{O}^{+}$, the sets of outcomes that are, respectively, more and less valuable than the corresponding expected outcome in the counterfactual scenario\footnote{Formally, $\mathcal{O}^{-} = \{o \in \mathcal{O}: E[V_{0}|O_{1}] < E[V_{1}|O_{1}]\}$ and $\mathcal{O}^{+} = \mathcal{O} - \mathcal{O}^{-}$.}.
\end{definition}

Restrictions \nameref{restriction:medium_information} and \nameref{restriction:high_information} present opposite extremes on the use of factual information that allows the groups in Definition \ref{def:outcome_separation} to be compensated differently.

\begin{restriction_}[M-FI]
  \label{restriction:medium_information}
  Medium Factual Information. The only information about the factual outcome that can be used is whether it is more or less valuable than the corresponding expected counterfactual outcome\footnote{Formally, $\mathcal{K} = \sigma(\{\mathcal{O}^{+}\})$.}. That is, while performing compensation valuation, one can use only the factual information of whether the outcome lies in $\mathcal{O}^{-}$ or $\mathcal{O}^{+}$. The compensation valuation must be the same for every outcome in $\mathcal{O}^{-}$ and also for every outcome in $\mathcal{O}^{+}$.
\end{restriction_}

\begin{restriction_}[H-FI]
  \label{restriction:high_information}
  High Factual Information. All the available information about the factual outcome can be used in the compensation valuation\footnote{Formally, $\mathcal{K}=\sigma(O_{1})$.}. That is, every factual outcome can lead to a different compensation (although not necessarily).
\end{restriction_}

While, Restriction \nameref{restriction:medium_information} uses the minimum factual information that is necessary so that outcomes in $\mathcal{O}^{-}$ and in $\mathcal{O}^{+}$ obtain different compensations, Restriction \nameref{restriction:high_information} allows one to use the complete factual information in order to determine the compensation. Hence, at the moment in which the tortious action occurs and considering all other factors fixed, Restrictions from \nameref{restriction:low_information} to \nameref{restriction:high_information} gradually increase both the uncertainty (at the time of the event) and the precision (at the time of the adjudication) of the future compensation.

Besides the use of factual information, Section \ref{section:elements} also doesn't completely specify what should be the connection between the factual and the counterfactual scenarios. Although we didn't find any previous discussion of this matter, we believe Restrictions \nameref{restriction:connection_evidence}, \nameref{restriction:connection_legal} and \nameref{restriction:connection_legal_2} provide reasonable ways to determine this connection.

\begin{restriction_}[E-C]
 \label{restriction:connection_evidence}
 Evidence Connection. The connection between the scenarios in Description \ref{description:no_choice} should be determined by evidence.
\end{restriction_}

\begin{restriction_}[LD-C]
 \label{restriction:connection_legal}
 Least Divergence Connection. The connection between the scenarios in Description \ref{description:no_choice} should be determined by law and should be chosen as the one that brings least divergence\footnote{Several restrictions in this Section use words such as ``divergence'', ``closeness'' and ``distance''. Here, and throughout, we use expected squared difference (euclidean distance) as a measure of distance between random quantities.} between the values obtained by the victim in each scenario\footnote{Formally, let $\mathcal{P}$ denote the set of all the joint distributions for $(F,O_{0},O_{1})$ that satisfy the conditions specified in Description \ref{description:no_choice} and in Figure \ref{figure:description_no_choice}. The joint distribution for $(F,O_{0},O_{1})$ should be taken according to the solution to $$\arg \min_{\mathbb{P} \in \mathcal{P}}{E_{\mathbb{P}}[(V_{0}-V_{1})^{2}]}.$$ Finding the minimum of the above equation is known as the Monge-Kantorovich mass transfer problem \citep{Kantorovich1960}.}. 
\end{restriction_}

\begin{restriction_}[I-C]
 \label{restriction:connection_legal_2}
 Independence Connection. The connection between the scenarios in Description \ref{description:no_choice} should be determined by law and should be chosen as the one that makes the scenarios independent\footnote{Formally, $\sigma(F) = \{\emptyset, \Omega\}$. That is, the probability model in Figure \ref{figure:description_no_choice} can be simplified to the following figure:
  \begin{align*}
   \xymatrix{
    *+[Fo]{O_{0}} \ar[d]	& *+[Fo]{O_{1}} \ar[d] \\
    *+[F-]{V_{0}}		& *+[F-]{V_{1}}
  }
 \end{align*}}.
\end{restriction_}

Examples \ref{example:factual-independence} and \ref{example:factual-dependence} can be used to contrast Restrictions \nameref{restriction:connection_evidence}, \nameref{restriction:connection_legal} and \nameref{restriction:connection_legal_2}. If one constructs the connection between scenarios based on evidence, then compensation in these two examples can be different. Evidentially, while in Example \ref{example:factual-independence} a blue ball draw in the factual scenario brings no information about the counterfactual draw, in Example \ref{example:factual-dependence} it brings certainty about this draw . As possible alternatives, Restrictions \nameref{restriction:connection_legal} and \nameref{restriction:connection_legal_2} treat Examples \ref{example:factual-independence} and \ref{example:factual-dependence} equivalently. In Example \ref{example:factual-independence}, Restriction \nameref{restriction:connection_legal} compares a blue draw from box 1 to a blue draw from box 2. That is, Restriction \nameref{restriction:connection_legal} generates a certain correspondence in both examples between a blue factual draw and a blue counterfactual draw. On the contrary, in Example \ref{example:factual-dependence}, Restriction \nameref{restriction:connection_legal_2} ignores the connection given by evidence between the factual and counterfactual scenarios. In summary, Restriction \nameref{restriction:connection_evidence} distinguishes the compensation in Examples \ref{example:factual-independence} and \ref{example:factual-dependence}. Restriction \nameref{restriction:connection_legal_2} models both examples as Restriction \nameref{restriction:connection_evidence} models Example \ref{example:factual-independence}, while Restriction \nameref{restriction:connection_legal} models both as Restriction \nameref{restriction:connection_evidence} models Example \ref{example:factual-dependence}.

The last conceptual question raised in the beginning in this Section was about the legal interpretation of indemnity in lost chance cases. One idea associated to indemnity is that it should ``undo'' the damage caused to the victim. Another idea that is associated to indemnity is that the factual scenario with compensation should have the same value as the counterfactual scenario without compensation.  Restrictions \nameref{def:indemnity_1} and \nameref{def:indemnity_2} stress the difference between these concepts in the context of the compensation for a lost chance.

\begin{restriction_}[CC-I]
 \label{def:indemnity_1}
 Closest to Counterfactual Indemnity. The compensation valuation should be the function in Definition \ref{def:compensation} that brings the victim the closest to the counterfactual scenario\footnote{Formally, $X = \arg\min_{\{X \in \mathcal{X}\}}{E[((V_{0}-V_{1})-X)^{2}]}$.}.
\end{restriction_}

\begin{theorem}
 \label{theorem:indemnity_1}
 Assume that there exists a compensation valuation, $X$, such that $E[(V_{0}-(V_{1}+X))^{2}] < \infty$ and that $E[|V_{0}-V_{1}|] < \infty$. There exists a unique\footnote{Formally, we say that there exists a unique random variable that satisfies a restriction if, for every random variables $X$ and $Y$ that satisfy the restriction, $P(X=Y)=1$.} compensation valuation, $X^{*}$, that satisfies Restriction \nameref{def:indemnity_1} and
 \begin{align*}
   X^{*}	&= \max(0,E[V_{0}-V_{1}|\mathcal{K}])
 \end{align*}
\end{theorem}

The proof of Theorem \ref{theorem:indemnity_1} is in the Appendix.

\begin{restriction_}[FM-I]
 \label{def:indemnity_2}
 Fixed Mean Indemnity.  Before the tortious actions occurs, the victim should be indifferent between a situation in which the tortious action doesn't occur and a situation in which it occurs and the victim receives compensation. The compensation valuation should be the function in Definition \ref{def:compensation} that satisfies the above restriction and that brings the victim the closest to the counterfactual scenario\footnote{Formally, $X = \arg\min_{\{X \in \mathcal{X}: E[X]=E[V_{0}-V_{1}] \text{ or } X=0\}}{E[((V_{0}-V_{1})-X)^{2}]}$.}.
\end{restriction_}

\begin{theorem}
 \label{theorem:indemnity_2}
 Assume that there exists a compensation valuation, $X$, such that $E[(V_{0}-(V_{1}+X))^{2}] < \infty$, $E[X]=E[V_{0}-V_{1}]$ and that $E[|V_{0}-V_{1}|] < \infty$. There exists a unique compensation valuation, $X^{*}$, that satisfies Restriction \nameref{def:indemnity_2}. If $E[V_{0}-V_{1}] \leq 0$, then $X^{*}=0$. If $E[V_{0}-V_{1}] > 0$, then there exists a unique $\lambda^{*} \in \mathbb{R}^{+}$ that satisfies the equation $E[\max(0,E[V_{0}-V_{1}|\mathcal{K}]-\lambda^{*})] = E[V_{0}-V_{1}]$ and $X^{*} = \max(0,E[V_{0}-V_{1}|\mathcal{K}]-\lambda^{*})$.
\end{theorem}

The proof of Theorem \ref{theorem:indemnity_2} is also in the Appendix.

Restrictions \nameref{def:indemnity_1} and \nameref{def:indemnity_2} lead to different compensations only when it is possible to obtain an outcome in the factual scenario with a larger value than the corresponding expected value in the counterfactual scenario. In case such an outcome occurred, a naive interpretation of ``undoing'' the tortious action would involve fining the victim. Restrictions \nameref{def:indemnity_1} and \nameref{def:indemnity_2} deal with the impossibility of this naive interpretation of ``undo''. Whenever the expected counterfactual outcome is better than the factual outcome, Restriction \nameref{def:indemnity_1} compensates the victim to bring him to the expected counterfactual outcome. In order obtain this effect, Restriction \nameref{def:indemnity_1} makes, at the time of the tortious action, the factual scenario with compensation to be preferable to the victim than the counterfactual scenario. Restriction \nameref{def:indemnity_2}, however, establishes that the factual scenario with compensation should have the same expected value as the counterfactual scenario. Hence, in order for some factual outcomes to be better than the corresponding expected counterfactual values, other factual outcomes (with compensation) must be worse. We use Example \ref{example:indemnity_difference} to illustrate the difference between Restrictions \nameref{def:indemnity_1} and \nameref{def:indemnity_2}.

\begin{example}
  \label{example:indemnity_difference}
  In the counterfactual scenario one out of $5$ prizes would be awarded to the victim with equal probability. Let $\mathcal{O} = \{a_{1},a_{2},a_{3},a_{4},a_{5}\}$ stand for the possible awards and $V(a_{1})=5,V(a_{2})=30,V(a_{3})=35,V(a_{4})=70,V(a_{5})=110$. Evidence establishes that the tortfeasor tampered with the probabilities of the awards in such a way that there is a deterministic relation between the factual and counterfactual scenarios that is given by Table \ref{table:indemnity_difference}.
\end{example}

\begin{table}[ht]
 \centering
 \begin{tabular}{|l|c|c|c|c|c|}
  \hline
  Variable & \multicolumn{5}{|c|}{Outcomes}											\\
  \hline
  $O_{0}$ 								& $a_{1}$ & $a_{2}$ & $a_{3}$ & $a_{4}$ & $a_{5}$	\\
  \hline
  $O_{1}$ (Restriction \nameref{restriction:connection_evidence})	& $a_{3}$ & $a_{3}$ & $a_{2}$ & $a_{1}$ & $a_{4}$	\\
  $O_{1}$ (Restriction \nameref{restriction:connection_legal})		& $a_{1}$ & $a_{2}$ & $a_{3}$ & $a_{4}$ & $a_{3}$	\\
  \hline
 \end{tabular}
 \caption{Possible outcomes in the counterfactual scenario $(O_{0})$ followed by their corresponding outcomes in the factual scenario $(O_{1})$.}
 \label{table:indemnity_difference}
\end{table}

If Restrictions \nameref{restriction:high_information} and \nameref{restriction:connection_evidence} are used, then there exists a difference between Restrictions \nameref{def:indemnity_1} and \nameref{def:indemnity_2} in Example \ref{example:indemnity_difference}. This difference occurs since the outcome $a_{3}$ in the factual scenario has a larger value than the corresponding expected value in the counterfactual scenario. Hence, if the outcome in the factual scenario is $a_{3}$, the victim obtains no compensation. Also, if $a_{1}$ occurs in the factual scenario, then one knows that $a_{4}$ would have happened in the counterfactual scenario. Hence, Restriction \nameref{def:indemnity_1} would award the victim $V(a_{4})-V(a_{1}) = 70-5 = 65$. However, Restriction \nameref{def:indemnity_2} would provide a smaller compensation value in order to balance out the fact that, when $a_{3}$ occurs, the victim obtains an outcome that is better than the expected counterfactual value. In order to obtain this effect, Restriction \nameref{def:indemnity_2} would award the victim a value of $50$ when $a_{1}$ occurs in the factual scenario. 

Table \ref{table:compensations} summarizes the compensation valuations obtained by combining different legal restrictions. Restrictions \nameref{restriction:connection_evidence}, \nameref{restriction:connection_legal} and \nameref{restriction:connection_legal_2} don't appear in the first row since they are implicitly taken into account by the formulas that are derived. In combination with either Restriction \nameref{def:indemnity_1} or \nameref{def:indemnity_2}, Restriction \nameref{restriction:medium_information} provides a larger compensation to outcomes in $\mathcal{O}^{+}$ than Restriction \nameref{restriction:low_information}. This occurs because Restriction \nameref{restriction:medium_information} allows the decision-maker not to compensate the victims that obtain a factual outcome better than the counterfactual outcome. As a consequence, the decision-maker can increase the compensation in the counterfactual outcomes in which the victim was disadvantaged by the tort. Similarly, by allowing complete information about the factual outcome, Restriction \nameref{restriction:high_information} allows compensation to be the closer than Restrictions \nameref{restriction:low_information} and \nameref{restriction:medium_information} to the trier of fact's current evaluation of damages. Finally, no matter whether Restrictions \nameref{restriction:low_information}, \nameref{restriction:medium_information} or \nameref{restriction:high_information} are used, the compensation valuation obtained with Restriction \nameref{def:indemnity_1} is always at least as high as the one obtained with Restriction \nameref{def:indemnity_2}.

\begin{table}
 \centering
 \begin{tabular}{|c|c|}
  \hline
  Restrictions											& Compensation Valuation										\\
  \hline
  \nameref{restriction:low_information} and (\nameref{def:indemnity_1} or \nameref{def:indemnity_2})	& $X = \max(0, E[V_{0}-V_{1}])$										\\[2mm]
  \nameref{restriction:medium_information} and \nameref{def:indemnity_1}				& $X = \mathcal{O}^{+} \cdot E[V_{0}-V_{1}|\mathcal{O}^{+}]$						\\[2mm]
  \nameref{restriction:medium_information} and \nameref{def:indemnity_2}				& $X = \mathcal{O}^{+} \cdot \frac{ \max(0,E[V_{0}-V_{1}])}{P(\mathcal{O}^{+})}$			\\[2mm]
  \nameref{restriction:high_information} and \nameref{def:indemnity_1}					& $X = \max(0, E[V_{0}-V_{1}|O_{1}])$									\\[2mm]
  \nameref{restriction:high_information} and \nameref{def:indemnity_2}					& $X = \max(0, E[V_{0}-V_{1}|O_{1}]-\lambda^{*})$ 							\\[1mm]
  \hline
 \end{tabular}
 \caption{Compensation valuations obtained by combining different legal restrictions}
 \label{table:compensations}
\end{table}

Despite the fact that the compensation valuations in Table \ref{table:compensations} might seem unfamiliar, they can be compared to two of the formulas\footnote{\citet{Noah2005} indicates that the formula $\frac{p_{0}-p_{1}}{p_{0}} \Delta v$ has also been used to calculate compensation. None of the compensation valuations we obtain in Table \ref{table:compensations} reduces to this formula in Example \ref{example:medical_malpractice}. As an intuitive argument against the formula $\frac{p_{0}-p_{1}}{p_{0}}\Delta v$, observe that it doesn't reduce to $p_{0} \Delta v$ when $p_{1}=0$. Instead, \citet{Noah2005} observes that, in this situation, $\frac{p_{0}-p_{1}}{p_{0}}\Delta v$ reduces to $\Delta v$. The author argues that this formula is unreasonable by pointing out that, if the chance of a better outcome were reduced from 2\% to 0\%, then the formula would prescribe the victim to receive the total difference between the values of the good and bad outcomes (instead of 2\% of this value, as prescribed by proportional damage).} presented in \citet{Noah2005} for medical malpractice cases (Example \ref{example:medical_malpractice}). Table \ref{table:compensation_malpractice_example} summarizes the compensation valuations for each possible combination of restrictions and outcome. From this Table one finds that, in medical malpractice cases, both the formulas $(p_{0}-p_{1})\Delta v$ and $\frac{p_{0}-p_{1}}{1-p_{1}}\Delta v$ can be justified according to the choice of Restrictions one uses. If Restrictions \nameref{restriction:connection_evidence} or \nameref{restriction:connection_legal} are used, then, while $(p_{0}-p_{1})\Delta v$ is justified when no factual information can be used to determine the compensation, $\frac{p_{0}-p_{1}}{1-p_{1}}\Delta v$ is obtained when this information can be used and one doesn't compensate the good outcomes. Notwithstanding the justification for both formulas, we find no justification for using the formula $(p_{0}-p_{1})\Delta v$ while compensating only the victims that obtain the bad outcome. Also, the formula obtained by combining Restrictions \nameref{restriction:connection_legal_2} and \nameref{def:indemnity_1} with the use of factual information is $p_{0} \Delta v$ and resembles the rule of proportional damage. In this case, the compensation is always the same no matter what is the reduction, $p_{0}-p_{1}$, in the probability of obtaining $o_{g}$, because Restriction \nameref{restriction:connection_legal_2} models the factual and counterfactual scenarios as independent.

\begin{table}
  \centering
  \begin{tabular}{|c|c|c|}
    \hline
																											& \multicolumn{2}{|c|}{Factual Outcome}																\\
    \hline
    Restrictions																									& $o_{b}$/red					& $o_{g}$/blue									\\
    \hline
    \nameref{restriction:low_information} and (\nameref{restriction:connection_evidence} or \nameref{restriction:connection_legal} or \nameref{restriction:connection_legal_2}) and (\nameref{def:indemnity_1} or \nameref{def:indemnity_2})	& $(p_{0}-p_{1})\Delta v$			& $(p_{0}-p_{1})\Delta v$	\\
    (\nameref{restriction:medium_information} or \nameref{restriction:high_information}) and (\nameref{restriction:connection_evidence} or \nameref{restriction:connection_legal}) and (\nameref{def:indemnity_1} or \nameref{def:indemnity_2}) & $\frac{p_{0}-p_{1}}{1-p_{1}} \Delta v$	& $0$				\\[1mm]
    (\nameref{restriction:medium_information} or \nameref{restriction:high_information}) and \nameref{restriction:connection_legal_2} and \nameref{def:indemnity_1}									& $p_{0} \Delta v$				& $0$				\\[1mm]
    (\nameref{restriction:medium_information} or \nameref{restriction:high_information}) and \nameref{restriction:connection_legal_2} and \nameref{def:indemnity_2}									& $\frac{p_{0}-p_{1}}{1-p_{1}} \Delta v$	& $0$				\\[1mm]
    \hline
  \end{tabular}
  \caption{Compensation valuation obtained in Examples \ref{example:medical_malpractice} and \ref{example:factual-dependence} for each possible outcome when the tortious action occurs and given different combinations of Restrictions.}
  \label{table:compensation_malpractice_example}
\end{table}

Although the compensation in Example \ref{example:medical_malpractice} is the same whether one applies Restriction \nameref{restriction:connection_evidence} or \nameref{restriction:connection_legal}, the pattern isn't general. Examples \ref{example:factual-independence} and \ref{example:factual-dependence} explore the difference between these restrictions. Example \ref{example:factual-dependence} is analogous to Example \ref{example:medical_malpractice} and, therefore, leads to the same compensation valuations as in Table \ref{table:compensation_malpractice_example}. However, in Example \ref{example:factual-independence}, Restriction \nameref{restriction:connection_evidence} generates the same model as Restriction \nameref{restriction:connection_legal_2} and, therefore, differs from Restriction \nameref{restriction:connection_legal}. Table \ref{table:compensation_connection_example} contrasts the compensation valuations obtained by applying different combinations of Restrictions in Example \ref{example:factual-independence}. If one uses Restriction \nameref{restriction:connection_legal}, then Examples \ref{example:factual-dependence} and \ref{example:factual-independence} share the same probabilistic model and compensations are equal. On the other hand, under Restriction \nameref{restriction:connection_evidence} or \nameref{restriction:connection_legal_2}, the counterfactual and factual outcomes are independent in Example \ref{example:factual-independence}. Consequently, under either of these restrictions the expected value of the counterfactual outcome that corresponds to a factual red ball is higher than under Restriction \nameref{restriction:connection_legal}. Hence, if the decision-maker allows the use of factual information to determine compensation and applies Restriction \nameref{def:indemnity_1}, he obtains a higher compensation valuation with Restrictions \nameref{restriction:connection_evidence} or \nameref{restriction:connection_legal_2} than with Restriction \nameref{restriction:connection_legal}. 

\begin{table}
  \centering
  \begin{tabular}{|c|c|c|}
    \hline
																												& \multicolumn{2}{|c|}{Factual Outcome}					\\
    \hline
    Restrictions																										& red					& blue				\\
    \hline
    \nameref{restriction:low_information} and (\nameref{restriction:connection_evidence} or \nameref{restriction:connection_legal} or \nameref{restriction:connection_legal_2}) and (\nameref{def:indemnity_1} or \nameref{def:indemnity_2})		& $(p_{0}-p_{1})\Delta v$		& $(p_{0}-p_{1})\Delta v$	\\
    (\nameref{restriction:medium_information} or \nameref{restriction:high_information}) and (\nameref{restriction:connection_evidence} or \nameref{restriction:connection_legal_2}) and \nameref{def:indemnity_1} 				 	& $p_{0} \Delta v$			& $0$				\\[1mm]
    (\nameref{restriction:medium_information} or \nameref{restriction:high_information}) and \nameref{restriction:connection_evidence} and \nameref{def:indemnity_2}			 							& $\frac{p_{0}-p_{1}}{1-p_{1}} \Delta v$& $0$				\\[1mm]
    (\nameref{restriction:medium_information} or \nameref{restriction:high_information}) and \nameref{restriction:connection_legal} and (\nameref{def:indemnity_1} or \nameref{def:indemnity_2}) 							& $\frac{p_{0}-p_{1}}{1-p_{1}} \Delta v$& $0$				\\[1mm]
    \hline
  \end{tabular}
  \caption{Compensation valuation obtained in Example \ref{example:factual-independence} for each possible outcome when the tortious action occurs and given different combinations of Restrictions.}
  \label{table:compensation_connection_example}
\end{table}

The differences between the combinations of Restrictions are further heightened when there are many possible outcomes, such as in Example \ref{example:indemnity_difference}. Table \ref{table:compensation_indemnity_example} presents the compensation valuations obtained in this example for each possible combination of Restrictions. In general, one can observe that the variability of compensations increases when the use of factual information increases from Restriction \nameref{restriction:low_information} to \nameref{restriction:high_information}. Also, for any fixed combination of use of factual information and connection between factual and counterfactual scenarios, the compensation under Restriction \nameref{def:indemnity_1} is always at least as high as under Restriction \nameref{def:indemnity_2}. Finally, while Restriction \nameref{restriction:connection_evidence} compares the factual and counterfactual scenarios based on evidence, Restrictions \nameref{restriction:connection_legal} and \nameref{restriction:connection_legal_2} use a connection defined by law. Restriction \nameref{restriction:connection_legal} focuses on the fact that the only difference between the possible outcomes in $O_{0}$ and $O_{1}$ (Table \ref{table:indemnity_difference}) is that $O_{1}$ substitutes $a_{5}$ by an extra possible occurrence of $a_{3}$. As a consequence, when factual information is allowed in combination with Restriction \nameref{restriction:connection_legal}, the victim is compensated only when $a_{3}$ is observed. Restriction \nameref{restriction:connection_legal_2} compares the factual value of each observed outcome with the average counterfactual value of $50$. The wide variety of compensation valuations that are found, by applying different combinations of Restrictions to Example \ref{example:indemnity_difference}, highlights the importance of the conceptual questions formulated at the beginning of this Section. 

\begin{table}
  \centering
  \begin{tabular}{|c|c|c|c|c|}
    \hline
																														& \multicolumn{4}{|c|}{Factual Outcome}				\\
    \hline																					
    Restrictions																												& $a_{1}$	& $a_{2}$	& $a_{3}$	& $a_{4}$	\\
    \hline
    \nameref{restriction:low_information} and (\nameref{restriction:connection_evidence} or \nameref{restriction:connection_legal} or \nameref{restriction:connection_legal_2}) and (\nameref{def:indemnity_1} or \nameref{def:indemnity_2})	& $15$		& $15$		& $15$		& $15$		\\
    \nameref{restriction:medium_information} and \nameref{restriction:connection_evidence} and \nameref{def:indemnity_1} 															& $36.6$	& $36.6$	& $0$		& $36.6$	\\
    \nameref{restriction:medium_information} and \nameref{restriction:connection_evidence} and \nameref{def:indemnity_2} 															& $25$		& $25$		& $0$		& $25$		\\
    \nameref{restriction:high_information} and \nameref{restriction:connection_evidence} and \nameref{def:indemnity_1} 																& $65$		& $5$		& $0$		& $40$		\\
    \nameref{restriction:high_information} and \nameref{restriction:connection_evidence} and \nameref{def:indemnity_2} 																& $50$		& $0$		& $0$		& $25$		\\
    (\nameref{restriction:medium_information} or \nameref{restriction:high_information}) and \nameref{restriction:connection_legal} and (\nameref{def:indemnity_1} or \nameref{def:indemnity_2})						& $0$		& $0$		& $37.5$	& $0$		\\
    \nameref{restriction:medium_information} and \nameref{restriction:connection_legal_2} and \nameref{def:indemnity_1} 															& $23.7$	& $23.7$	& $23.7$	& $0$		\\
    \nameref{restriction:medium_information} and \nameref{restriction:connection_legal_2} and \nameref{def:indemnity_2} 															& $18.7$	& $18.7$	& $18.7$	& $0$		\\
    \nameref{restriction:high_information} and \nameref{restriction:connection_legal_2} and \nameref{def:indemnity_1} 																& $45$		& $20$		& $15$		& $0$		\\
    \nameref{restriction:high_information} and \nameref{restriction:connection_legal_2} and \nameref{def:indemnity_2} 																& $40$		& $15$		& $10$		& $0$		\\
    \hline
  \end{tabular}
  \caption{Compensation valuation obtained in Example \ref{example:indemnity_difference} when the tortious action occurs and given different combinations of Restrictions.}
  \label{table:compensation_indemnity_example}
\end{table}

Next, we discuss the valuation of lost chances in which the victim can make a choice that affects his outcome. The main question in this type of situation is how to model what the victim would have chosen in the counterfactual scenario. We consider this question while building upon the results in this section.

\section{Valuation of lost chances with choices}
\label{section:choice}

Consider that, in the lost chance description, the victim is able to make a choice that affects his outcome. For example, under appropriate medical diagnosis, a victim might be able to choose among different treatments for his condition. Each choice influences the victim's final outcome in different ways. Description \ref{description:choice} provides an abstraction of this type of situation.

\begin{description_}
 \label{description:choice}
 The trier of fact determines that, in the following order, the tortfeasor committed a tortious action, the victim made a choice, $c$, among a set of choices $\mathcal{C}$, and among all relevant possibilities in a set $\mathcal{R}$, a result, $r$, occurred. This is the factual scenario. Furthermore, in the counterfactual scenario, had the tortious action not occurred, the choices in $\mathcal{C}$ and the results in $\mathcal{R}$ would occur with different probabilities. The trier of fact also assigns a value from the victim's perspective, $V(c,r) \in \mathbb{R}$, and a value from the tortfeasor's perspective, $V^{*}(c,r) \in \mathbb{R}$, to each combination of choice $c \in \mathcal{C}$ and result $r \in \mathcal{R}$, assigns for each value $v \in \mathbb{R}$ the amount of money $M(v)$ and $M^{*}(v)$ that has value $v$ for, respectively, the victim and tortfeasor and decides on what information, $\mathcal{K}$, about the result can be used to calculate compensation. Finally, the judge decides that it was the victim's duty to make a choice in a set $\mathcal{D} \subset \mathcal{C}$.
\end{description_}

From the perspective of the trier of fact, the choices in $\mathcal{C}$ and the results in $\mathcal{R}$ are similar with respect to their uncertainty. Usually, the trier of fact can't know for sure what the counterfactual choice or result would be. Given this uncertainty, we assume that the trier of fact assigns probabilities to both the victim's choices and the results. From this perspective, the model in Section \ref{section:elements} is applicable to Description \ref{description:choice}, using a space of outcomes that includes both the victim's choice and the result\footnote{Formally, for each $o \in \mathcal{O}$, $o=(c,r)$, where $c \in \mathcal{C}$ and $r \in \mathcal{R}$.}. This model is illustrated in Figure \ref{figure:description_choice_1}. In this Figure, the relevant outcomes are separated into two types of variables: while $C_{0}$ and $C_{1}$ represent, respectively, the counterfactual and factual choices of the victim, $R_{0}$ and $R_{1}$ represent, respectively, the counterfactual and factual results. One obtains Figure \ref{figure:description_no_choice} by calling $O_{0}=(C_{0},R_{0})$ and $O_{1}=(C_{1},R_{1})$. The arrows that connect to $V_{0}$ and $V_{1}$ indicate that the valuation of the victim's counterfactual and factual situations can depend on both the victim's choice and the final outcome, that is, $V_{0} = V(C_{0},R_{0})$ and $V_{1} = V(C_{1},R_{1})$. 

\begin{figure}[ht]
 \centering
 \begin{subfigure}{.4\textwidth}
  \begin{align*}
   \xymatrix{
    *+[Fo]{C_{0}} \ar[d] \ar@/_1pc/[dd] & *+[Fo]{F} \ar[dl] \ar[dr] \ar[l] \ar[r]	& *+[Fo]{C_{1}} \ar[d] \ar@/^1pc/[dd]	\\
    *+[Fo]{R_{0}} \ar[d]		& 						& *+[Fo]{R_{1}} \ar[d]			\\
    *+[F-]{V_{0}}			& 						& *+[F-]{V_{1},}
   }
  \end{align*}
  \caption{no restrictions.}
  \label{figure:description_choice_1}
 \end{subfigure}
 \begin{subfigure}{.4\textwidth}
  \begin{align*}
   \xymatrix{
    *+[Fo]{C_{0}} \ar[d] \ar@/_1pc/[dd] & *+[Fo]{F_{\mathcal{C}}} \ar[l] \ar[r]	& *+[Fo]{C_{1}} \ar[d] \ar@/^1pc/[dd]	\\
    *+[Fo]{R_{0}} \ar[d]		& *+[Fo]{F_{\mathcal{R}}} \ar[l] \ar[r]	& *+[Fo]{R_{1}} \ar[d]			\\
    *+[F-]{V_{0}}			& 					& *+[F-]{V_{1}}
   }
  \end{align*}
  \caption{$F_{\mathcal{R}}$ doesn't directly influence $F_{\mathcal{C}}$.}
  \label{figure:description_choice_2}
 \end{subfigure}

 \caption{Influence diagrams for Description \ref{description:choice}.}
 \label{figure:description_choice}
\end{figure}
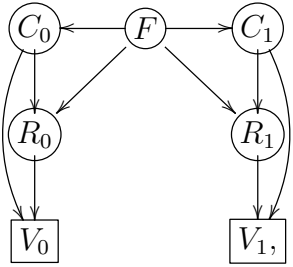
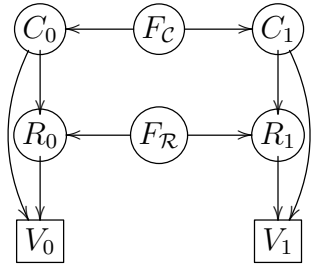

Despite the formal similarity between choices and outcomes, the separation in Figure \ref{figure:description_choice} is important due to a practical matter. It is often hard for either party to argue about what the victim would have done had the tortious action not occurred. Therefore, one can expect that the probabilities on $F$ and $C_{0}$ will often depend on legal presumptions and on which party has the burden of proof regarding these variables. Also, some of the choices available to the victim might be wrongful actions. This criterion wasn't used to classify the outcomes in Section \ref{section:no_choice}. In this Section, we discuss presumptions and consequences that are applicable to the counterfactual choices. Similarly to Section \ref{section:no_choice}, our analysis of the probability model used by the trier of fact relies on the answers given to three conceptual questions:

\begin{itemize}
 \item At the time the victim made his choice, had the tortious action not occurred, how much information about the result would the victim have available?
 \item Had the tortious action not occurred, how would the victim use the available information about the result in order to make a choice?
 \item The victim's choices are classified by whether they are wrongful or not. How does this classification affect the valuation of lost chances?
\end{itemize}

The importance of the first question can be understood by contrasting the probability models in Figures \ref{figure:description_choice_1} and \ref{figure:description_choice_2}. In Figure \ref{figure:description_choice_1} there is a single connection between the factual and counterfactual scenario. That is, the underlying variables that connect the factual and counterfactual choices are the same as the ones that connect the factual and counterfactual results. The sharing of the underlying variables allows the hypothesis that, at the time the victim made his choice, he had more information about the underlying variables that affect the result than that which was presented as evidence to the trier of fact. By contrast, in Figure \ref{figure:description_choice_2}, the connection between the counterfactual and factual situations, $F$, is separated into two variables: $F_{\mathcal{C}}$ and $F_{\mathcal{R}}$. While $F_{\mathcal{C}}$ is the connection between factual and counterfactual choices, $F_{\mathcal{R}}$ is the connection between the results. The lack of an arrow connecting $F_{\mathcal{R}}$ to $F_{\mathcal{C}}$ assumes that the underlying variables that connect the counterfactual and factual results don't directly influence the ones that connect the counterfactual and factual choices. Therefore, as opposed to Figure \ref{figure:description_choice_1}, Figure \ref{figure:description_choice_2} assumes that, at the time the victim made his decision, he didn't have access to the underlying variables that connect the factual and counterfactual results. Example \ref{example:treatment_choice} illustrates the difference between Figures \ref{figure:description_choice_1} and \ref{figure:description_choice_2}.

\begin{example}
 \label{example:treatment_choice}
 Consider that the individuals of a population can be classified in types $1$ and $2$. If an individual is of type $1$, then treatment $1$ is more effective than treatment $2$ to cure a given disease, $D$. Similarly, if an individual is of type $2$, then treatment $2$ is more effective than treatment $1$ to cure $D$. On average, both treatments are equally effective to cure an individual of an unknown type. In order to find a person's type, one must perform an exam that takes a large amount of time.
 
 In an emergency situation, a physician is responsible for treating a patient infected by $D$. Although the physician had the legal obligation of consulting with the patient about which treatment to apply, negligently and unaware of the victim's type, he applied treatment $1$ without communicating with the victim.
\end{example}

If either the model in Figure \ref{figure:description_choice_1} or \ref{figure:description_choice_2} were applied to Example \ref{example:treatment_choice}, then the patient's type would be an underlying variable that affects the result. If the model in Figure \ref{figure:description_choice_1} is used, then the patient's type could also affect the patient's decision regarding what treatment to use. On the contrary, the model in Figure \ref{figure:description_choice_2} doesn't allow this underlying variable to affect both choice and result. Figure \ref{figure:description_choice_2} allows two scenarios. In the first scenario, the patient isn't able to prove that, at the time he made his choice, he was aware of his type. Therefore, the model assumes that the patient couldn't act as if his type were known. In the second scenario, the patient is able to prove that he was aware of his type. In this case, the trier of fact should model the victim's type as a known fact, instead of as an underlying variable.

By restricting the use of the underlying variables by the victim, the probability model in Figure \ref{figure:description_choice_2} limits the way in which evidence about the factual choice and result influences the respective counterfactuals. Given that the factual choice is known, the factual result brings no information about the counterfactual choice. In this sense, the factual result influences the counterfactual choice only by bringing information about the factual choice. This limitation can be justified on the grounds that, if the victim had access to information about the underlying variables at the time he made his choice, then it would be the victim's burden to prove this information existed. The above presumption is described in Restriction \ref{restriction:choice_conditional_independence}. 

\begin{restriction_}[VK]
 \label{restriction:choice_conditional_independence}
 Victim's Knowledge. The victim should be presumed to know at each time only the facts that were proven in court to be known by victim at that time\footnote{Formally, the joint distribution for $(C_{0},C_{1},F_{\mathcal{C}},R_{0},R_{1},F_{\mathcal{R}})$ should follow the conditional independencies in Figure \ref{figure:description_choice_2}.}.
\end{restriction_} 

Besides the question of how much information was available to the victim, the trier of fact should also evaluate how the victim would use this information in order to make a choice. One position is that, in the lack of evidence, the judge should presume that the victim would make the best valued choice. Since the victim could lawfully choose any of the options, it is usually unreasonable to assume he would choose an undervalued option. Restriction \ref{restriction:choice_evidence} presents a method for obtaining the presumed choice.

\begin{restriction_}[IT-CP]
 \label{restriction:choice_evidence}
 \textit{Iuris Tantum} Choice Presumption. In the lack of evidence about how a party would have made a choice in Description \ref{description:choice}, the judge should assume that the party would certainly select the choice that is obtained following the steps\footnote{Formally, let $c^{*}_{0}=\arg \min_{c \in \mathcal{D}}{E[V_{0}|C_{0}=c]}$. Presume $P(C_{0}=c^{*}_{0})=1$.}:
 \begin{enumerate}
  \item The judge should consider an hypothetical situation in which the judge knows the same information as the party at the time the party made his choice.
  \item In the hypothetical situation, the judge picks the dutiful choice that would maximize the expected value obtained by the party according to the probability model assessed by the trier of facts.
 \end{enumerate}
\end{restriction_}

One can also defend that the choice obtained by Restriction \ref{restriction:choice_evidence} should always be presumed, no matter what evidence is available about the counterfactual choice. From this perspective, indemnification should be rewarded according to the value of the choice itself, and not according to how the victim would have used the choice. In favor of this argument, consider the example in \citet{Jansen1999} in which a person is robbed before going to a casino. The expected loss of money in the casino shouldn't be taken into account while calculating indemnification. By analogy, the possible uses of the lost choice shouldn't be taken into account while calculating indemnification. The value of a choice is usually taken to be the highest expected value obtainable over all possible options. Restriction \ref{restriction:choice_legal} presents this position.

\begin{restriction_}[II-CP]
 \label{restriction:choice_legal}
 \textit{Iuris et de Iure} Choice Presumption. No matter what evidence is presented, the counterfactual choice in Description \ref{description:choice} should always be presumed to be the option obtained according to the steps outlined in Restriction \ref{restriction:choice_evidence}.
\end{restriction_}

Although Restrictions \ref{restriction:choice_evidence} and \ref{restriction:choice_legal} present the same presumption about the counterfactual choice of the victim, the restrictions differ with respect to the force of this presumption. Restriction \ref{restriction:choice_evidence} presents a relative presumption, which may be rebutted by evidence. By contrast, Restriction \ref{restriction:choice_legal} presents an absolute presumption, which cannot be rebutted by evidence.

The last conceptual question we consider is about how the valuation of damages is affected by the victim's duty in Description \ref{description:choice}. Often, the victim might have the duty to take reasonable care to mitigate the damage caused by the tortfeasor. For example, consider a case in which a tenant vacates a rented building before the end of his lease. The landlord has the duty to mitigate the damage by trying to rent the building to a new tenant. In this type of case, the victim's breach of duty is a tortious action and, therefore, we can imagine a new hypothetical case in which the roles of victim and tortfeasor are exchanged. We say that the tortfeasor and victim in Description \ref{description:choice} are, respectively, the dual victim and the dual tortfeasor in Description \ref{description:choice_2}.

\begin{description_}
 \label{description:choice_2}
 The trier of fact determines that the dual tortfeasor had a duty to make a decision in a set $\mathcal{D} \subset \mathcal{C}$, but the dual tortfeasor neglected his duty and chose $c \notin \mathcal{D}$. Among all relevant possibilities in a set $\mathcal{R}$, a result, $r$, occurred. This is the factual scenario. Furthermore, in the counterfactual scenario, had the dual tortfeasor performed his duty and chosen $c \in \mathcal{D}$, the results in $\mathcal{R}$ would occur with different probabilities. The trier of fact also assigns a value from the dual victim's perspective, $V^{*}(c,r) \in \mathbb{R}$ to each combination of choice $c \in \mathcal{C}$ and result $r \in \mathcal{R}$, assigns for each value $v \in \mathbb{R}$ the monetary value $M(v) \in \mathbb{R}$ that has value $v$ to the dual victim and decides on what information, $\mathcal{K}$, about the result can be used to calculate compensation.
\end{description_}

The probability model for Description \ref{description:choice_2} is presented in Figure \ref{figure:mitigation_individual}. The variables $C_{1}$ and $R_{1}$ are the same as in Figure \ref{figure:description_choice_2}. The variables $C_{1,0}$ and $R_{1,0}$ are, respectively, the choice of the dual tortfeasor and it's result, under the assumption that the tortfeasor acted wrongfully and the dual tortfeasor performed his duty. The indemnification in Description \ref{description:choice_2} is computed according to the dual victim's values, that is, $V^{*}_{1} = V^{*}(C_{1},R_{1})$ and $V^{*}_{1,0} = V^{*}(C_{1,0},R_{1,0})$. Figure \ref{figure:mitigation_joint} presents the full influence diagram obtained by considering simultaneously Descriptions \ref{description:choice} and \ref{description:choice_2}.

In case the victim in Description \ref{description:choice} neglected his duty, then he is a dual tortfeasor such as in Description \ref{description:choice_2}. Consequently, the monetary compensation obtained by the victim from the tortfeasor in the lost chance case in Description \ref{description:choice} should be reduced by the monetary amount that the victim owes to the tortfeasor in the lost chance case in Description \ref{description:choice_2}. This understanding conforms to the practice of reducing the compensation of a victim who doesn't take reasonable precautions to mitigate damages. Restriction \ref{restriction:mitigation} presents this position.

\begin{figure}[ht]
 \centering
 \begin{subfigure}{.4\textwidth}
  \begin{align*}
   \xymatrix{
    *+[Fo]{C_{1,0}} \ar[d] \ar@/_1pc/[dd] & *+[Fo]{F^{*}_{\mathcal{C}}} \ar[l] \ar[r]	& *+[Fo]{C_{1}} \ar[d] \ar@/^1pc/[dd]	\\
    *+[Fo]{R_{1,0}} \ar[d]		& *+[Fo]{F^{*}_{\mathcal{R}}} \ar[l] \ar[r]	& *+[Fo]{R_{1}} \ar[d]			\\
    *+[F-]{V^{*}_{1,0}}			& 						& *+[F-]{V_{1}^{*}}
   }
  \end{align*}
  \caption{Influence diagram for Description \ref{description:choice_2}.}
  \label{figure:mitigation_individual}
 \end{subfigure}
 \begin{subfigure}{.4\textwidth}
  \begin{align*}
   \xymatrix{
    *+[Fo]{C_{0}} \ar[d] \ar@/_1pc/[dd] & *+[Fo]{F_{\mathcal{C}}} \ar[l] \ar[r]	& *+[Fo]{C_{1}} \ar[d] \ar@/_/[ddl] \ar@/^/[ddr]& *+[Fo]{F_{\mathcal{C}}^{*}} \ar[l] \ar[r]	& *+[Fo]{C_{1,0}} \ar[d] \ar@/^1pc/[dd]	\\
    *+[Fo]{R_{0}} \ar[d]		& *+[Fo]{F_{\mathcal{R}}} \ar[l] \ar[r]	& *+[Fo]{R_{1}} \ar[dl]	\ar[dr]			& *+[Fo]{F_{\mathcal{R}}^{*}} \ar[l] \ar[r]	& *+[Fo]{R_{1,0}} \ar[d]		\\
    *+[F-]{V_{0}}			& *+[F-]{V_{1}}				& 						& *+[F-]{V_{1}^{*}}				& *+[F-]{V_{1,0}^{*}}
   }
  \end{align*}
  \caption{Influence diagram for Descriptions \ref{description:choice} and \ref{description:choice_2}.}
  \label{figure:mitigation_joint}
 \end{subfigure}
 \caption{Influence diagrams related to Description \ref{description:choice_2}.}
 \label{figure:mitigation}
\end{figure}
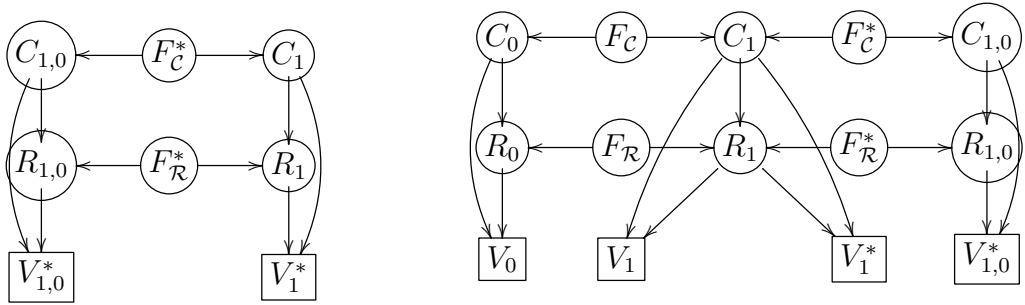

\begin{restriction_}
 \label{restriction:mitigation}
 In case the victim in Description \ref{description:choice} neglects his duty, then his compensation monetary value should be subtracted by the compensation monetary value due to the tortfeasor because of this breach in duty\footnote{Formally, let $X_{0}$ and $X_{1,0}$ be, respectively, the compensations obtained in Figures \ref{figure:description_choice_2} and in \ref{figure:mitigation_individual}. The monetary value awarded to the victim should be $\max(0,M(V_{1}+X_{0})-M(V_{1})-M^{*}(V^{*}_{1}+X_{1,0})+M^{*}(V^{*}_{1}))$. For example, if the judge adopts Restriction \nameref{def:indemnity_1}, then the monetary value awarded to the victim should be $\max(0,\max(0,M(E[V_{0}|\mathcal{K}])-M(E[V_{1}|\mathcal{K}]))-\max(0,M(E[V^{*}_{1,0}|\mathcal{K}])-M(E[V^{*}_{1}|\mathcal{K}])))$.}.
\end{restriction_}

By combining the restrictions in Sections \ref{section:no_choice} and \ref{section:choice}, one obtains compensation policies for lost chance cases with choices. In Section \ref{section:choice_example} we apply these compensation policies to a lost chance case in which the victim had the right to a choice, \citet{Matos2005}.

\section{An application of the valuation of lost chances with choices}
\label{section:choice_example}

An important precedent in the application of the Lost Chance Doctrine in Brazil is \citet{Matos2005}, known as the ``Millionaire Show'' case. Similar to its American version, the ``Millionaire Show'' was a Brazilian television show in which in each round a guest could answer a multiple choice question. If the guest's answer was correct, her cumulative prize would increase and she would pass to the next round. If the guest's answer was incorrect she would be awarded a fraction of her cumulative prize and her participation in the show would end. Finally, if the guest decided not to answer, she would get the full extent of her cumulative prize and her participation would end.

The guest in 06/15/2000, Matos, got to the last round of the show. Her cumulative prize was $R\$500,000$. The last multiple choice question which was presented to her had four options. The question asked the proportion of Brazil's territory that the country's Constitution reserves to natives: 22\%, 2\%, 4\% or 10\%? Had Matos answered correctly, she would get $R\$1,000,000$. Had she answered incorrectly, she would get $R\$300$. She decided not to answer and keep the cumulative prize of R\$500,000. 

Nevertheless, it was later found out that, although the plaintiff had registered one of the alternative answers as correct, none of them were actually correct. The Brazilian Constitution doesn't directly specify a percentage of Brazil's territory reserved to natives, but specifies that natives have the right over the land that they had traditionally occupied\footnote{\citet{Brasil1988} specifies: ``art. 231. s\~{a}o reconhecidos aos \'{i}ndios sua organiza\c{c}\~{a}o social, costumes, l\'{i}nguas, cren\c{c}as e tradi\c{c}\~{o}es, e os direitos origin\'{a}rios sobre as terras que tradicionalmente ocupam, competindo \`{a} Uni\~{a}o demarc\'{a}-las, proteger e fazer respeitar todos os seus bens''.}. Based on this error, Matos filed a lawsuit against the producers of the show. Her case reached the Superior Court of Justice, the highest appellate court in Brazil for non-constitutional matters. It was the first case the Court decided using the Lost Chance Doctrine. Although the Court ruled that Matos wasn't able to prove that she would get the $R\$1,000,000$ prize if faced with a proper question, it also decided that Matos had been wrongfully deprived of her right to answer a properly formulated question and that the producers of the show were liable for this damage. The Court stated that the ``mathematical probability'' of correctly answering a multiple choice question with four options was $25\%$ and, therefore, Matos should be compensated with 25\% of the difference between the total prize and the R\$500,000 that she had already received. Next, we compare the Court's valuation with the method we propose in Sections \ref{section:no_choice} and \ref{section:choice}. 

We start by exposing the elements of Description \ref{description:choice} that are present in \citet{Matos2005}. The first element in Description \ref{description:choice} is the set of choices that the victim could choose from, $\mathcal{C}$. After the multiple choice question was presented to Matos, she could choose between two options: either to answer or not to answer the question. That is, $\mathcal{C} = \{\text{answer},\text{not answer}\}$. The next element in Description \ref{description:choice} is the set of possible results of the victim's choice, $\mathcal{R}$. In case Matos chose to answer, she would receive either $R\$1,000,000$ or $R\$300$. In case Matos chose not to answer, then she would certainly receive $R\$500,000$. Therefore, the set of all relevant results that could follow from Matos's choice was $\mathcal{R} = \{R\$300, R\$500,000, R\$1,000,000\}$

Given the previous two elements, we proceed to the relevant probabilities for \citet{Matos2005}. No matter if the tortious action occurred or not, Matos would certainly receive $\$500,000$ in case she chose not to answer\footnote{Formally, $P(R_{0}=\$500,000|C_{0}=\text{answer}) = P(R_{1}=\$500,000|C_{1}=\text{not answer}) = 1$.}. By contrast, in case Matos had answered the question, then two outcomes could occur. In the factual scenario, the alternative that had been marked as correct by the plaintiff was actually incorrect. In this case, Matos's best hope to obtain the answer marked as correct was to randomly guess with equal probability one of the alternatives\footnote{Formally, $P(R_{1}=R\$1,000,000|C_{1}=\text{answer}) = 0.25$.}. In the counterfactual scenario, one should consider a question such that the answer marked as correct was also a correct answer to the question that was asked. In this case, Matos would have been able to rely on her knowledge in order to find the correct answer.

In this respect, we disagree with the Court's decision. Matos's probability of correctly answering a properly designed question should depend, at least, on Matos's skill and the difficulty of the question. Probability theory does not specify that the ``mathematical probability'' of correctly answering to a multiple choice question with four alternatives is 0.25. On the contrary, this probability should be judged according to the evidence that was available. For example, the alternatives that were presented to Matos ($22\%$, $2\%$, $4\%$ and $10\%$) were qualitatively different. In a properly formulated question, qualitative differences can be used to rule out some of the alternatives (is it reasonable to expect more than a fifth of Brazil's territory to be reserved to natives?). Using this method and then guessing, Matos would improve her probability from 25\%. Furthermore, Matos had shown skill by correctly answering all questions that had previously been asked in the show. This evidence also gives reason to believe that Matos would probably be able to choose the correct alternative with a greater probability than if she were to randomly guess one of the alternatives with equal probability. 

By contrast, the defendant claimed that the questions' difficulties progressed as the contestant got closer to the final prize. For example, among the first questions that were asked to previous contestants, there was ``what fruit is dried to obtain dried plums? 1) plum, 2) grape, 3) peach, 4) melon''. By contrast, the defendant presented evidence that, until that date, no contestant had successfully answered the last question in the show. Therefore, the defendant argued that the high probability obtained by considering Matos's record of successful answers should be balanced by the higher difficulty of the last question.

A possible alternative to the Court's evaluation of the probability was to subjectively balance the evidence presented by the victim and the defendant. In the following, we aim for a general analysis of Matos's probability of success. We refer to the probability that Matos would choose the correct alternative of a properly formulated question by $P(R_{0}=R\$1,000,000|C_{0}=\text{answer})$ and describe the consequences of all possible choices for this probability on the valuation of the lost chance.

The next elements in Description \ref{description:choice} are the victim's utility for each outcome, $V$, and the amount of money that is necessary to increase the victim's utility by $v$, $M(v)$. In principle, $V$ could depend on both the choice, $c$, and the result, $r$. Nevertheless, Matos claimed a lost chance over only the monetary gain. For example, Matos didn't claim a loss over the excitement of answering a properly formulated question. Therefore, we assume that the value of each outcome is a function of the money that is obtained in that outcome. Using this simplification, $V$ returns a valuation of money and $M$ returns the amount of money that corresponds to a given value. That is, $M$ is a function of $V$\footnote{Formally, we assume that $V$ is a monotone increasing function and that $M(v) = V^{-1}(v)$.} and, therefore, it is sufficient to specify $V$ alone.

In \citet{Matos2005}, the Court didn't directly specify the victim's utility for each outcome. Instead, the Court directly assigned that the compensation should be 25\% of a difference between monetary amounts. By using this rule, the Court implicitly assumed the identity between amount of money and value. One consequence of this assumption is that the victim would value equally the options of receiving R\$500,000 with certainty and receiving R\$1,000,000 with probability 50\% and nothing otherwise. It might be unreasonable to believe a person has this type of preference. For example, it would be a common answer to prefer the certain R\$500,000 to the risky R\$1,000,000, a phenomenon known as risk aversion. 

Risk aversion is such a commonly observed phenomenon that we consider the effect of other valuations of money on the compensation. We value money with a class of functions\footnote{Formally, $V(m) = \frac{1-m^{1-\theta}}{\theta-1}$, for $0 \leq \theta \leq 1$. In case $\theta = 0$, utility is linear and there is no risk aversion. Risk aversion increases with $\theta$ and one obtains $V(m) = \log(m)$ as $\theta$ goes to $1$.}, presented in \citet{Kadane2011}, that has a parameter that permits adjustment to the amount of risk aversion of the victim. The trier of fact's judgment on this parameter could be informed, for example, by the fact that Matos decided to answer all the questions before the last one. Based on the questions' difficulties and Matos's skill, the trier of fact would be able to set a bound for Matos's risk aversion.

The last element of Description \ref{description:choice} defines the set of choices the victim had the duty to choose from, $\mathcal{D}$. That is, we evaluate whether Matos took appropriate measures to mitigate her damages. This element wasn't discussed by either the parties or the Court. Since Matos chose the safe option and avoided answering the improperly formulated question that was asked to her, we believe her factual choice was dutiful\footnote{Formally, $\mathcal{D} \subset \{\text{answer}\}$.}. 

Given the previously defined elements, the Restrictions in Sections \ref{section:no_choice} and \ref{section:choice} completely specify the lost chance valuation. In Matos's case, the factual result was a certain consequence of her chance. As a consequence, all the combinations of the different types of Restrictions provide the same valuation.

Figure \ref{fig:compensation_matos_1} illustrates the compensation that should be awarded to Matos under the two extreme cases of risk aversion. The left graph represents the case in which Matos has no risk aversion. In this case, Matos would need at least a probability of $\frac{500,000-300}{1,000,000-300}$ (approximately $50\%$) of correctly answering the question in order to obtain some compensation. Intuitively, if this probability were lower than $\frac{500,000-300}{1,000,000-300}$, then, in the counterfactual scenario, not answering would be the best choice for Matos. After the probability of $\frac{500,000-300}{1,000,000-300}$ is attained, Matos's compensation increases linearly until it reaches $R\$500,000$ when her probability of correctly answering is $1$. The right graph represents the case in which Matos has a proportion of risk aversion of $1$. In this case, Matos would need at least a probability of $\frac{\log(500,000)-\log(300)}{\log(1,000,000)-\log(300)}$ (approximately 91.5\%) to be compensated. This value is higher than in the previous case because, under risk aversion, Matos would take the risky chance of answering the question only if she were extremely confident that she would answer correctly. Similarly to the previous graph, compensation increases with Matos's probability of success and reaches $R\$500,000$ when this probability is $1$.

\begin{figure}
  \centering
  \includegraphics[scale=0.75]{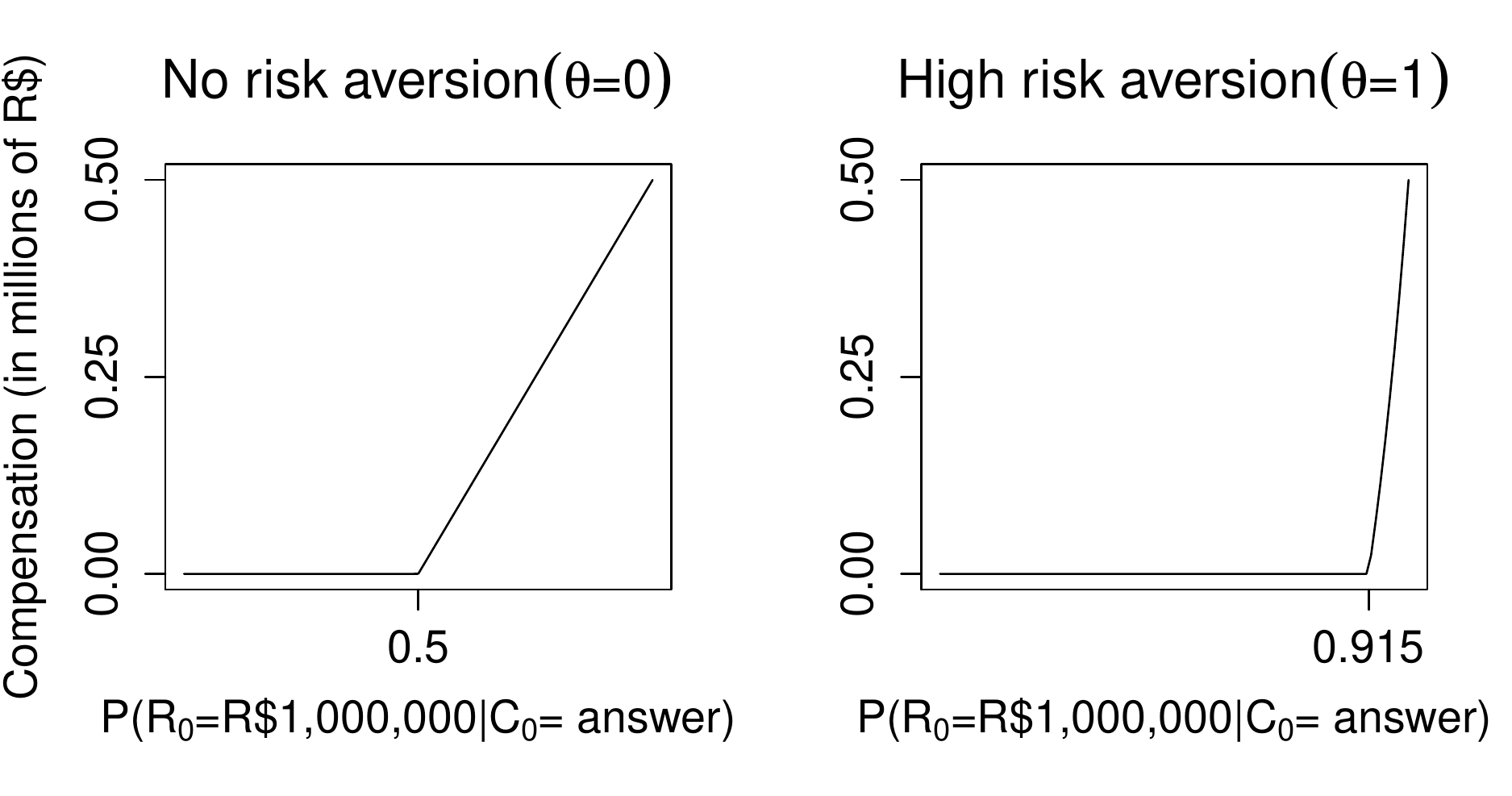}
  \caption{Compensation that should be awarded in \citet{Matos2005} as a function of Matos's probability of correctly answering the question.}
  \label{fig:compensation_matos_1}
\end{figure}

Figure \ref{fig:compensation_matos_2} illustrates the variation of Matos's compensation in terms of her risk aversion and probability of correctly answering the question. While the white coloring indicates points such that Matos should be awarded no compensation, the grey coloring indicates those such that Matos should be awarded some compensation. One can observe that, for a fixed probability of success, Matos's compensation decreases as her risk aversion increases. This effect occurs because, as Matos becomes more averse to risk, the more she prefers the safe option of receiving R\$500,000 to the risky option of competing for R\$1,000,000. The border between the two lightest grey regions marks the points such that Matos receives R\$125,000, the amount that was awarded to her by the Court. If we use no risk aversion, as the Court did, Matos would require a probability of 62.5\% in order to be awarded R\$125,000. This value is much higher than the probability of 25\% that was used by the Court. Similarly, if Matos's risk aversion parameter were $1$, then she would need a probability of $94.2\%$ to get her award to $R\$125,000$.

\begin{figure}[t]
  \centering
  \includegraphics{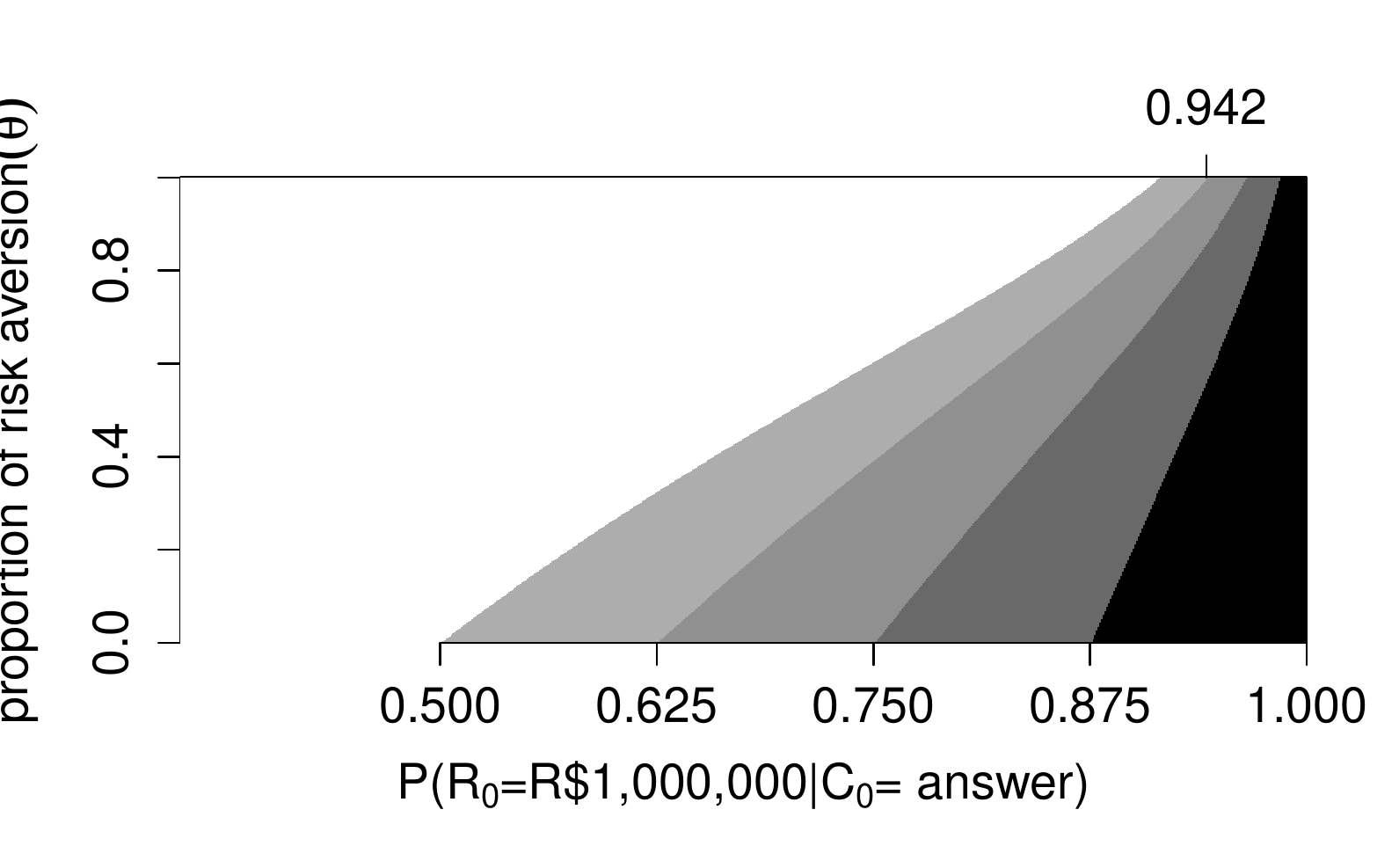}
  \caption{Illustration of the compensation that should be awarded in \citet{Matos2005} as a function of Matos's risk aversion and Matos's probability of correctly answering to a proper question. The white region indicates the points such that no compensation is awarded. The colors are ordered from light gray to dark gray and indicate compensation according to the following intervals: $(R\$0;R\$125,000]$, $(R\$125,000;R\$250,000], (R\$250,000;R\$375,000]$, $(R\$375,000;R\$500,000]$.}
  \label{fig:compensation_matos_2}
\end{figure}

The above analysis exposes an argument that wasn't directly recognized by the Court. Matos's lost chance involved a possibility of losing her cumulative prize. Consequently, the lost chance might have been worse than the factual outcome, that is, the analogy between Matos's case and the description in the rule of proportional damage fails. By using this analogy, the Court avoided explicitly weighting the importance of elements such as the rules of the television show, Matos's skill, Matos's aversion to risk and the question's difficulty. These elements are key figures in our analysis and could change the compensation to any value between $R\$0$ and $R\$500,000$.

\section{Conclusion}
\label{sec:conclusion}

What legal restrictions inform the valuation of lost chances? Despite the practical importance of this question, it has seldom been discussed explicitly. On the contrary, the discussion around the valuation of lost chances has mostly revolved around the formula prescribed by the rule of proportional damage. 

The above focus on a single formula is dangerous for at least two reasons. The first reason is that the formula applies only to a very specific type of case. Furthermore, when slight variations of this specific case (such as medical misdiagnosis) are considered, the different formulas that have been proposed can yield widely different valuations. Without a justification for choosing a given formula, the compensation value that is obtained is arbitrary. The second reason is that the formulas that have been discussed require the trier of fact to specify only one or two probabilities and a single monetary value. As a consequence, they induce the trier of fact to disregard elements that don't easily fall into these categories. For example, in \citet{Matos2005}, the Court ignored relevant questions such as Matos's skill, the difficulty of the questions, Matos's risk aversion and the different possible outcomes of the lost chance. 

Based on the above considerations, we propose a discussion of the legal restrictions that guide the valuation of lost chances. We frame this discussion in terms of six conceptual questions. 

The first three questions can be applied to every lost chance case: How much information about the factual outcome can be used in the valuation of lost chances? How does the law connect the scenario that was observed to the scenario that would have occurred but for the tortious action? What is the meaning of indemnification for a lost chance? Despite the qualitative nature of these questions, we show that the answers that are given to them completely specify compensation valuation in a description that is more general than the one required for the application of the rule of proportional damage.

We explore alternative combinations of answers to the above questions using two types of examples. As a first example, we use the typical case of medical misdiagnosis. In this case, we show in Table \ref{table:compensation_malpractice_example} that the main formulas that have been proposed in the literature can be obtained from specific answers. Consequently, the differences between these qualitative answers can be used to understand the proposed formulas and make justified choices between them. A second example, consists of a though experiment that involves a lost chance with $5$ possible outcomes. As a summary, Table \ref{table:compensation_indemnity_example} shows that different combinations of answers to the conceptual questions can lead to radically different compensation valuations. Some of these compensations disagree completely even with respect to what outcomes should be compensated. This disagreement illustrates the necessity of discussing the conceptual questions that were raised.

The other three conceptual questions we present apply to cases in which the victim had the right to a choice that could affect the probabilities of his possible outcomes: At the time the victim made his choice, how much information did he have available? How would the victim use this information to make a choice? Is a victim's compensation affected by a wrongful choice? Since it is hard to predict the behavior of a human, both of the parties might fail to answer the first two questions. In this case, we discuss possible legal presumptions.

We use the six conceptual questions to study \citet{Matos2005}. Despite this case being Brazil's Supreme Court of Justice's first application of the lost chance doctrine, in this case the victim had multiple possible outcomes and also had the right to make a choice. The conceptual questions we propose raise relevant elements for compensation that weren't explicitly considered by the Court's decision.

\section*{Acknowledgement}

The authors are grateful for the insightful discussions with Celso Campilongo, Andr\'{e} Katsurada, Caroline Mitchell, Marcelo Nunes, Carlos Pereira, Ted Permigiani, Fernando Rozenblit, Julio Stern and Julio Trecenti.

\bibliographystyle{lpr}
\bibliography{LostChance}

\vspace{-5mm}
\section*{Appendix}
\begin{definition}
  In the following proofs, we use:
  \begin{enumerate}
   \item $R(X) = E[(V_{0}-(V_{1}+X))^{2}]$.
   \item $\Delta V_{\mathcal{K}} = E[V_{0}-V_{1}|\mathcal{K}]$.
  \end{enumerate}
\end{definition}

\begin{lemma}
 \label{lemma:l2_trick}
 If $E[|V_{0}-V_{1}|] < \infty$, then for every $X, X^{*} \in \mathcal{X}$ such that $R(X) < \infty$ and $R(X^{*}) < \infty$,
 \begin{align*}
  R(X)-R(X^{*}) &= E[(X-\Delta V_{\mathcal{K}})^{2}] - E[(X^{*}-\Delta V_{\mathcal{K}})^{2}]
 \end{align*}
\end{lemma}

\begin{proof} For every $X \in \mathcal{X}$,
  \begin{align*}
  R(X)		&= E[E[((V_{0}-V_{1})-X)^{2}|\mathcal{K}]]						&												\\
		&= E[((V_{0}-V_{1})-\Delta V_{\mathcal{K}})^{2}] + E[(\Delta V_{\mathcal{K}}-X)^{2}]	& \text{since $E[|V_{0}-V_{1}|] < \infty$, $E[\Delta V_{\mathcal{K}}] = E[V_{0}-V_{1}]$}	\\
		&											& \text{and $X$ and $\Delta V_{\mathcal{K}}$ are $\mathcal{K}$-measurable}
 \end{align*}
 Hence,
 \begin{align*}
  R(X)-R(X^{*})	&=	E[(\Delta V_{\mathcal{K}}-X)^{2}] - E[(\Delta V_{\mathcal{K}}-X^{*})^{2}]
 \end{align*}
\end{proof}

\begin{proof}[Proof of Theorem \ref{theorem:indemnity_1}]  Let $X^{*}_{1} = \max(0, \Delta V_{\mathcal{K}})$. For every $X \in \mathcal{X}$,
 \begin{align*}
  R(X)-R(X^{*}_{1})	&= E[(\Delta V_{\mathcal{K}}-X)^{2}] - E[(\Delta V_{\mathcal{K}}-X^{*}_{1})^{2}]												& \text{Lemma \ref{lemma:l2_trick}}							\\
			&=	E[((\Delta V_{\mathcal{K}}-X)^{2}-(\Delta V_{\mathcal{K}}-X^{*}_{1})^{2}) \cdot \mathcal{O}^{+}] + E[(X^{2}-2\Delta V_{\mathcal{K}}X) \cdot \mathcal{O}^{-}]	& \text{$X^{*}_{1}\cdot \mathcal{O}^{-} = 0$}						\\
			&\geq	E[((\Delta V_{\mathcal{K}}-X)^{2}-(\Delta V_{\mathcal{K}}-X^{*}_{1})^{2}) \cdot \mathcal{O}^{+}] + E[(X-0)^{2} \cdot \mathcal{O}^{-}]				& \text{$\Delta V_{\mathcal{K}}\cdot \mathcal{O}^{-} \leq 0$, $X \geq 0$}		\\
			&= E[(\Delta V_{\mathcal{K}}-X)^{2} \cdot \mathcal{O}^{+}] + E[(X-0)^{2} \cdot \mathcal{O}^{-}]										& \text{$X^{*}_{1}\cdot \mathcal{O}^{+}= \Delta V_{\mathcal{K}} \cdot \mathcal{O}^{+}$}	\\
			&= E[(X-X^{*}_{1})^{2}]
 \end{align*}
 Hence, if $P(X \neq X^{*}_{1}) > 0$, $R(X) > R(X^{*}_{1})$.
\end{proof}

\begin{lemma}
 \label{lemma:unique_solution}
 If $E[|V_{0}-V_{1}|] < \infty$ and $E[V_{0}-V_{1}] > 0$, then there exists a unique $\lambda^{*} \in \mathbb{R}$ such that
 \begin{align*}
  E[\max(0,\Delta V_{\mathcal{K}}-\lambda^{*})] = E[V_{0}-V_{1}]
 \end{align*}
\end{lemma}

\begin{proof}
 Let $f(\lambda) = E[\max(0,\Delta V_{\mathcal{K}}-\lambda)]$. Observe that $f$ is a decreasing function and that, for every $\epsilon > 0$ and $\lambda \in \mathbb{R}$, $|f(\lambda+\epsilon)-f(\lambda)| \leq \epsilon$. Hence, $f$ is a continuous function. Furthermore, 
 $$f(0) = E[\max(0, \Delta V_{\mathcal{K}})] \geq E[\Delta V_{\mathcal{K}}] = E[V_{0}-V_{1}]$$
 and, since $E[|V_{0}-V_{1}|] < \infty$, $E[|\Delta V_{\mathcal{K}}|] < \infty$ and
 $$\lim_{\lambda \rightarrow \infty}f(\lambda) \leq \lim_{\lambda \rightarrow \infty}\int_{\lambda}^{\infty}{x dF_{\Delta V_{\mathcal{K}}}(x)} = 0$$
 Since $f$ is continuous, $f(0) \geq E[V_{0}-V_{1}]$ and $\lim_{\lambda \rightarrow \infty}{f(\lambda)} = 0$ conclude from the intermediate value theorem that there exists a $\lambda^{*}$ such that $f(\lambda^{*}) = E[V_{0}-V_{1}]$. Since $f$ is decreasing, conclude that $\lambda^{*}$ is the only $\lambda \in \mathbb{R}$ such that $f(\lambda) = E[V_{0}-V_{1}]$.
\end{proof}

\begin{lemma}
 \label{lemma:variance_decomposition}
 Let $\lambda \in \mathbb{R}$, $X^{+}_{2} = \max(0, \Delta V_{\mathcal{K}}-\lambda)$ and $X^{-}_{2} = \min(0, \Delta V_{\mathcal{K}}-\lambda)$. For every $X \in \mathcal{X}$,
 \begin{align*}
  Var[X-\Delta V_{\mathcal{K}}]	&= Var[X-X^{+}_{2}]-2Cov[X,X^{-}_{2}]+2Cov[X^{+}_{2},X^{-}_{2}]+Var[X^{-}_{2}]
 \end{align*}
\end{lemma}

\begin{proof}
 \begin{align*}
  Var[X-\Delta V_{\mathcal{K}}]	&= Var[X-(\Delta V_{\mathcal{K}}-\lambda)]									\\
				&= Var[X] - 2Cov[X, \Delta V_{\mathcal{K}}-\lambda] + Var[V_{\mathcal{K}}-\lambda]				\\
				&= Var[X] - 2Cov[X, X^{+}_{2}+X^{-}_{2}] + Var[X^{+}_{2}+X^{-}_{2}]						\\
				&= Var[X] - 2Cov[X, X^{+}_{2}] -2Cov[X, X^{-}_{2}] + Var[X^{+}_{2}] + Var[X^{-}_{2}] + 2Cov[X^{+}_{2},X^{-}_{2}]	\\
				&= Var[X-X^{+}_{2}]-2Cov[X,X^{-}_{2}]+2Cov[X^{+}_{2},X^{-}_{2}]+Var[X^{-}_{2}]
 \end{align*}
\end{proof}

\begin{proof}[Proof of Theorem \ref{theorem:indemnity_2}] Let $\mathcal{X}_{1}$ = $\{X \in \mathcal{X}: E[X] = E[V_{0}-V_{1}] \text{ or } X=0\}$. If $X \in \mathcal{X}$, then $X \geq 0$. Therefore, if $E[V_{0}-V_{1}] \leq 0$, $\mathcal{X}_{1} = \{X=0\}$ and $\arg \min_{\{X \in \mathcal{X}_{1}\}}R(X) = 0$. Next, consider that $E[V_{0}-V_{1}] > 0$. It follows from Lemma \ref{lemma:unique_solution} that there exists a unique $\lambda^{*} \in \mathbb{R}$ such that $E[\max(0,\Delta V_{\mathcal{K}}-\lambda^{*})]=E[V_{0}-V_{1}]$. Let $X^{+}_{2}$ be defined as $\max(0, \Delta V_{\mathcal{K}}-\lambda^{*})$. It remains to show that $\arg \min_{X \in \mathcal{X}_{1}}{R(X)} = X^{+}_{2}$.

 For every $X \in \mathcal{X}_{1}$,
 \begin{align*}
  R(X)-R(X^{+}_{2})	&=	E[(\Delta V_{\mathcal{K}}-X)^{2}] - E[(\Delta V_{\mathcal{K}}-X^{+}_{2})^{2}]	& \text{Lemma \ref{lemma:l2_trick}}			\\
			&=	Var[\Delta V_{\mathcal{K}}-X] - Var[\Delta V_{\mathcal{K}}-X^{+}_{2}]		& \text{$E[X] = E[X^{+}_{2}]$}				\\
			&=	Var[X-X^{+}_{2}] -2Cov[X,X^{-}_{2}] +2Cov[X^{+}_{2},X^{-}_{2}]			& \text{Lemma \ref{lemma:variance_decomposition}}	\\
			&=	Var[X-X^{+}_{2}] -2(E[X \cdot X^{-}_{2}] - E[X^{+}_{2} \cdot X^{-}_{2}])	& E[X] = E[X^{+}_{2}]					\\
			&= 	Var[X-X^{+}_{2}] -2 E[X \cdot 	X^{-}_{2}]					& X^{+}_{2} \cdot X^{-}_{2} = 0				\\
			&\geq	Var[X-X^{+}_{2}]								& X \cdot X^{-}_{2} \leq 0
 \end{align*}
 Also,
 \begin{align*}
  R(0)-R(X^{+}_{2})	&=	E[\Delta V_{\mathcal{K}}^{2}] - E[(\Delta V_{\mathcal{K}}-X^{+}_{2})^{2}]	& \text{Lemma \ref{lemma:l2_trick}}			\\
			&= 	E[(X^{+}_{2}+X^{-}_{2})^{2}] - E[(X^{-}_{2})^{2}]				&							\\
			&=	E[(X^{+}_{2})^{2}]								& X^{+}_{2} \cdot X^{-}_{2} = 0				\\
			&=	E[X^{+}_{2}]^{2} + Var[X^{+}_{2}]						&							\\
			&>	E[\Delta V_{\mathcal{K}}]^{2} > 0						& E[\Delta V_{\mathcal{K}}] = E[V_{0}-V_{1}] > 0,	\\
			&											& E[X^{+}_{2}]>E[\Delta V_{\mathcal{K}}]
 \end{align*}

 Hence, if $E[V_{0}-V_{1}] > 0$, then for every $X \in \mathcal{X}_{1}$ such that $P(X \neq X^{+}_{2}) > 0$, $R(X) > R(X_{2}^{+})$.
\end{proof}

\end{document}